\documentclass[a4paper,11pt]{article}

\newtheorem{theorem}{Theorem}
\newtheorem{lemma}{Lemma}
\newtheorem{definition}{Definition}

\newtheorem{corollary}{Corollary}
\newenvironment{proof}{\noindent {\bf Proof:\,\ }}{\hfill\mbox{$\Box$}\smallskip}

\usepackage{lineno}
\usepackage{graphicx}
\usepackage{wrapfig}
\usepackage{amsmath,amsfonts,amssymb}
\usepackage{pifont}
\usepackage[utf8]{inputenc}
\usepackage{hyperref}
\usepackage{graphicx}
\usepackage{siunitx}
\usepackage{subfigure,multirow}
\usepackage{mathtools}
\usepackage{wrapfig}
\usepackage{array}
\usepackage{booktabs}
\newcommand{\scsi}{\mathsf{SCSI}}
\newcommand{\scss}{\mathsf{SCSS}}
\newcommand{\scsc}{\mathsf{SCSC}}
\newcommand{\scsr}{\mathsf{SCSR}}
\newcommand{\scst}{\mathsf{SCST}}
\newcommand{\pdia}{\mathsf{pair\_diagram}}
\newcommand{\reg}{\mathsf{reg}}
\renewcommand{\to}{\mathsf{Type~\text{-}~I}}
\newcommand{\tyt}{\mathsf{Type~\text{-}~II}}

\usepackage{microtype}
\usepackage[x11names,dvipsnames,table]{xcolor} 
\newtheorem{observation}{Observation}

\newcommand{\remove}[1]{}

\newcommand{\blue}[1]{{\textcolor{blue}{#1}}}
\usepackage{amsmath}
\usepackage{mdframed}
\usepackage{lipsum}

\newmdenv[
backgroundcolor=gray!15,
topline=false,
bottomline=false,
skipabove=\topsep,
skipbelow=\topsep,
leftmargin=-5pt,
rightmargin=-5pt,
innertopmargin=3pt,
innerbottommargin=3pt
]{siderules}

\usepackage{hyperref}
\hypersetup{
	colorlinks,
	citecolor=magenta,
	linkcolor=cyan,
	pdfstartview={FitH}
}

\title{Color spanning Localized query \footnote{A preliminary 
version of the paper  appeared   in the proceedings of 5th International Conference on Algorithms and Discrete Applied Mathematics, CALDAM 2019.}}

\author{Ankush Acharyya\thanks{Indian Statistical Institute, Kolkata, India.
	Email: {\tt \{ankush\_r, nandysc\}@isical.ac.in}.} \and
		Anil Maheshwari\thanks{Carleton University, Ottawa, Canada.
		Email:~{\tt anil@scs.carleton.ca}.}  \thanks{Research supported by NSERC.}\and Subhas C. Nandy\footnotemark[2]}
\date{\vspace{-5ex}}

\begin{document}
\maketitle

\begin{abstract}
Let $P$ be a set of $n$ points and each of the points is colored with one of the $k$ 
possible colors. We present efficient algorithms to pre-process $P$ such that for 
a given query point $q$, we can quickly identify the smallest color spanning object of the desired type
containing $q$. In this paper, we focus on $(i)$  
intervals, (ii) axis-parallel square, $(iii)$ 
axis-parallel rectangle, $(iv)$ equilateral triangle of fixed orientation and $(v)$ circle, as our desired type of objects.
\end{abstract}

	
\section{Introduction}
\blue{\em Facility location} is a widely studied problem in modern theoretical computer science. Here we have different types of facilities $\cal F$, and each facility 
has multiple copies, the distance is measured according to a given distance metric $\delta$, and the target is to optimize certain objective function by choosing the 
{\em best} facility depending on the problem. In \blue{\em localized query} problems, the objective is to identify the {\em best} facility containing 
a given query point $q$. In this paper, we consider the following version of the localized facility location problems.
\begin{siderules}
\noindent Given a set ${\cal P}=\{p_1,p_2,\ldots,p_n\}$ of $n$ points, each point is colored with one of the $k$ possible colors. 
The objective is  to pre-process the point set $\cal P$ such that given a query point $q$, we can quickly identify a color spanning object of the desired type that contains $q$ 
having optimum size. 
\end{siderules}
The motivation of this problem stems from the {\em color spanning} variation of the facility location problems 
\cite{abellanas2001smallest,das2009smallest,hasheminejad2015computing,jiang2014shortest,khanteimouri2013computing}. 
Here facilities of type $i \in \{1,2,\ldots,k\}$ is modeled as points of color $i$, and the objective is to identify 
the locality of desired geometric shape containing at least one point of each color such that the desired measure parameter 
(width, perimeter, area etc.) is optimized. Abellanas et al. \cite{abellanas2001smallest} mentioned algorithms for finding 
the smallest color spanning axis parallel rectangle and arbitrary oriented strip in $O(n(n-k)\log^2k)$  and 
$O(n^2\alpha(k)\log k + n^2\log n)$ time, respectively. These results were later improved by Das et al. \cite{das2009smallest} 
to run in $O(n(n-k)\log k)$ and $O(n^2\log n)$ time, respectively.  Abellanas et al. also pointed out that the smallest color 
spanning circle can be computed in $O(nk\log (nk))$ time using the technique of computing the upper envelope of Voronoi surfaces 
as suggested by Huttenlocher et al. \cite{huttenlocher1993upper}. The smallest color spanning  interval can be 
computed  in $O(n)$ time for an ordered point set of size $n$ with $k$ colors on a real line \cite{khanteimouri2013spanning}. 
Color spanning $2$-interval\footnote{given a colored point set on a line, the goal is to find two intervals that cover at least one 
point of each color such that the maximum length of the intervals is minimized}, axis-parallel equilateral triangle, 
axis-parallel square  can be computed in $O(n^2)$ \cite{jiang2014shortest}, $O(n\log n)$ \cite{hasheminejad2015computing}, 
$O(n\log^2 n)$ \cite{khanteimouri2013computing} time respectively.

The problem of constructing an efficient data structure with a given set $P$ of $n$ (uncolored) points 
to find an object of optimum (maximum/minimum) size (depending on the nature of the problem) 
containing a query point $q$ is also studied recently in the literature. These are referred to as the {\it localized query problem}. 

Given a point set $P$ of $n$ points and a query point $q$, Augustine et al. \cite{augustine2010recognizing} considered the problem of 
identifying the largest empty circle and largest empty rectangle that contains $q$ but does not contain any point from $P$. The 
query to identify a (i) largest empty circle takes $O(\log n)$ time, using a data structure
build in $O(n^2)$ space and $O(n^2\log n)$ time. (ii) largest empty rectangle takes $O(\log n)$ time, using a data structure
build in $O(n^2 \log n)$ time and space. Later, Augustine at al. \cite{augustine2013localized} considered a simple
polygon $P$ with $n$ vertices. They proposed a data structure build in $O(n\log n)$ time and $O(n\log ^2n)$ space,
which is an improvement over their previous data structure \cite{augustine2010recognizing} build in $O(n\log ^3n)$
time for the same problem. This improved data structure reports the largest empty circle
containing the query point $q$, without containing any vertices of $P$ in $O(\log n)$ time.
They also proposed two improved algorithms for the largest empty circle problem for point set $P$. They proposed two data structures:
build in $O(n^{\frac{3}{2}}\log^2 n)$ time and $O(n^{\frac{3}{2}}\log n)$ space and in $O(n^{\frac{5}{2}}\log n)$ time and $O(n^{\frac{5}{2}})$
space, respectively. There corresponding query time complexities are $O(\log n \log \log n)$ and $O(\log n)$, respectively. 
Kaplan et. al \cite{kaplan2011finding} improved the complexity results for both the variations of localized query 
problem. The query version of the largest empty rectangle problem can be 
reported in $O(\log ^4 n)$ time using a $O(n\alpha(n)\log^3 n)$ space data structure build in $(O(n\alpha(n)\log^4 n)$ time. For the query
version of the largest empty circle problem, the query time is $O(\log ^2n)$ time. The corresponding data structure can be build in 
$(O(n\log n)$ time and $O(n\log^2 n)$ space respectively. Constrained versions of the largest empty circle problem have also been studied by
Augustine et al. \cite{augustine2010largest}. When the center is constrained to lie on a horizontal query line $\ell$, the largest empty circle
can be reported in $O(\log n)$ time. The preprocessing time and space for the data structure is $O(n\alpha(n)\log n)$. They also considered
the version when the center is constrained to lie on a horizontal query line $\ell$ passing through a fixed point. Using a data structure build
in $O(n\alpha(n)^{O(\alpha(n))}\log n)$ time and space, the largest empty circle can be identified in $O(\log n)$ time. Both the data structures
use the well known $3$D lifting technique. Later Kaminker et al. addressed a special case of the largest empty circle recognition problem, where
for a given set of $O(n)$ circles they can be preprocessed into a data structure of $O(n\log ^3n)$ time and $O(n)$ space which answers 
the query in $O(\log n)$ time. This technique can also be extended to find the largest empty circle for a point set $P$ containing the query point
using the same time and space complexities.  Gester et al. \cite{gester2015largest} studied the problem of finding the largest empty square inscribed in a 
rectilinear polygon $\cal R$ containing a query point. For $\cal R$ without holes, a data structure build in $O(n)$ time and space can answer the
query in $O(\log n)$ time. For the polygon $\cal R$ containing holes, they used a data structure build in $O(n)$ time  and $O(n\log n)$ space. The
query answering can be done in $O(\log n)$ time. Their results can also be extended to find the largest empty square within a point set containing
a query point.

\begin{table}[h!]
\begin{small}
\begin{center}
    \caption{Complexity results for variations of color spanning objects containing $q$.} 
    \label{restable}
    \vspace{0.2in}
\begin{tabular}{|c|c|c|c|} 
\hline 
\rowcolor{blue!10}&\multicolumn{2}{c|}{ {\bf \blue{Pre-processing}}}&\\ 

\rowcolor{blue!10}\multirow{-2}{*}{{\bf \blue{Object}}}& \cellcolor{blue!30}{\em Time }& \cellcolor{blue!30} {\em Space} & \multirow{-2}{*}{{\bf \blue{Query Time}}}\\
\hline 
$\scsi$ & $O(n\log n)$ & $O(n)$ & $O(\log n)$ \\ 
\rowcolor{gray!20}  $\scss$  & $O\Big(N(\frac{\log N}{\log \log N})^2\Big)$, $N=\Theta(nk)$ & $ O(N\log ^2N) $ & $O\Big((\frac{\log N}{\log \log N})^3\Big)$ \\
$\scsr$ & $O(N\log N)$, $N=\Theta\big((n-k)^2\big)$ & $ O\Big(N(\frac{\log N}{\log \log N})\Big)$   & $O\Big((\frac{\log N}{\log \log N})^2\Big)$  \\
\rowcolor{gray!20} $\scst$ & $O(n\log n)$ & $O\Big(n(\frac{\log n}{\log \log n})\Big)$ & $O\Big((\frac{\log n}{\log \log n})^2\Big)$ \\  
$\scsc$   & \multirow{2}{*}{$O(n^2k\log n)$}&\multirow{2}{*}{$O\Big(n^2(\alpha(n))\Big)$} & \multirow{2}{*}{$O(\log n)$} \\
$(1+\epsilon)$-approximation &                    &                    &    \\
\bottomrule
\end{tabular}
\end{center}
\end{small}
\end{table}

\subsection{Our Contribution:}
In this paper, we study the localized query variations of different color spanning objects among a set $P=P_1\cup P_2\cup \ldots \cup P_k$ of points where $P_i$ 
is the set of points of color $i$, $|P_i|=n_i$, $\sum_{i=1}^k n_i = n$. The desired (color-spanning) objects are (i) smallest interval $(\scsi)$,  (ii)
 smallest axis-parallel square $(\scss)$, (iii) smallest axis-parallel rectangle $(\scsr)$, (iv) smallest equilateral triangle of fixed orientation $(\scst)$, and (v) smallest circle $(\scsc)$. For problem (i), the points are 
distributed in $I\!\!R$, and for all other problems, the points are placed in $I\!\!R^2$. For problem (v), we present an $(1+\epsilon)$ approximate solution. To solve this general version of the problem we optimally solve a constrained version of the problem where the center of the circle is constrained to lie on a given line $\ell$ and contains the query point $q$. \remove{To the best of our knowledge this localized query version of the problem has never been studied before and with an application in facility location problem, we believe this version of the problem is interesting.} The results are summarized in Table \ref{restable}. 



\section{Preliminaries}\label{prelim}\vspace{-0.1in}
We use $x(p)$ and $y(p)$ to denote the $x$- and $y$-coordinates of a point $p$. We use  
$dist(.,.)$ to define  the Euclidean distance between (i) two points, (ii)  
two lines, or (iii) a point and a line depending on the context. We use $h(\sigma)$ and $v(\sigma)$
to denote the horizontal and vertical lines through $\sigma \in \mathbb{R}^2$ respectively. 
Without loss of generality, we assume that all the points in $P$ are in the positive quadrant.

\subsection{Range Searching:}
A range tree $\cal R$ is a data structure that supports counting and reporting the points of a given set 
$P \in \mathbb{R}^d$ ($d \geq 2$) that lie in an axis-parallel rectangular query range $R$. $\cal R$ can 
be constructed in $O(n\log^{d-1} n)$ time and space where the counting and reporting for the query range 
$R$ can be done in $O(\log^{d-1} n)$ and $O(\log^{d-1} n +k)$ time, respectively, where $k$ is the number 
of reported points. Later, JaJa et al. \cite{jaja2004space} proposed a dominance query data structure in 
$\mathbb{R}^d$, that supports $d (\geq 3)$ dimensional query range $[\alpha_1, \infty] \times [\alpha_2, 
\infty] \times\ldots\times [\alpha_d, \infty]$ with the following result.

\begin{siderules}
 \begin{lemma}\cite{jaja2004space}\label{jaja}
 Given a set of $n$ points in $\mathbb{R}^d~(d\geq 3)$, it can be preprocessed in a data structure 
 $\cal T$ of size $O\Big(n(\frac{\log n}{\log \log n})^{d-2}\Big)$ in $O(n\log ^{d-2}n)$  time  
 which supports counting points inside a query rectangle in $O\Big((\frac{\log n}{\log \log n})^{d-1}\Big)$ time.
\end{lemma}
\end{siderules}

We have computed all possible color-spanning objects of desired type $\Psi$ (e.g., interval, corridor, 
square, rectangle, triangle) present on the plane, stored those objects in a data structure, and when 
the query point $q$ appears, we use those objects to compute the smallest object containing $q$. The 
optimum color spanning object of type $\Psi$ containing $q$ may be any one of the following two types:
\begin{description}
 \item[$\to$:] $q$  lies inside the optimum color-spanning object. Here, we search the smallest one among 
 the pre-computed (during preprocessing) color spanning objects 
 containing $q$.
 \item[$\tyt$:] $q$ lies on the boundary of the optimum color-spanning object. This is obtained by 
 expanding some pre-computed object to contain $q$. 
 \end{description}
 Comparing the best feasible solutions of $\to$ and $\tyt$, the optimum solution is returned. 
 

\section{Smallest color spanning intervals $(\scsi)$}\vspace{-0.1in}
We first consider the basic one-dimensional version of the problem: 
\begin{siderules}
Given a set of $P$ of $n$ 
points on a real line $L$, where each point has one of the $k$ possible colors; for a query point 
$q$ on the line report the smallest color spanning interval containing $q$, efficiently.   
\end{siderules}

Without loss of generality assume that $L$ is horizontal. We sort the points in $P$ from left to 
right, and for each point, we can find the minimal color-spanning interval in total $O(n)$ time 
\cite{chen2013algorithms}. Next, we compute the following three lists for each point $p_i \in P$. 
\begin{description}
 \item[start:]  the smallest color spanning interval starting at $p_i$, denoted by $\mathbf{start}(p_i)$.
 \item[end:]  the smallest color spanning interval ending at $p_i$, denoted by $\mathbf{end}(p_i)$
 \item[span:]  the smallest color spanning interval spanning $p_i$ (here $p_i$ is not the starting or 
 ending point of the interval), denoted by $\mathbf{span}(p_i)$.
\end{description}
\remove{
 Let the interval between $p_i$ and $p_{i+1}$ be $I_i$. While computing the color spanning
 intervals for each point, we also store the smallest color spanning interval information
 corresponding to the closure of each interval $I_{{cs}_i}$, $i=1, 2, \ldots, n-1$. This
 can be easily obtained using the color spanning interval information ({\em start, end, span}
 lists) corresponding to the points.  Now, we query with the point $q$ to find the intervals those contain $q$. 
}
Observe that, none of the minimal color spanning intervals corresponding to the points $p_i \in P$
can be a proper sub-interval of the color spanning interval of some other point \cite{chen2013algorithms}.

\begin{lemma}\label{pre-scsi}	 
 For a given point set $P$, the aforesaid data structure takes $O(n)$ amount of space and can be computed
 in $O(n\log n)$ time.
\end{lemma}
\remove{
\begin{proof}
 The color-spanning interval algorithm described in \cite{chen2013algorithms} can be extended to build the
 lists {\em start, end, span} corresponding to each point and the set of color spanning interval as well as
 for each interval. All the lists are of size $O(n)$. Thus the total data structure takes linear amount of
 space and can be build in O(n) time.
\end{proof}
}

During the query (with a point $q$), we identify $p_i,p_{i+1} \in P$ such that $x(p_i) \leq x(q)\leq x(p_{i+1})$. Next, 
we compute the lengths $OPT_C$ and $OPT_E$ of the  $\to$  $\scsi$
and $\tyt$ $\scsi$, respectively, as follows. Finally, we compute $OPT=\min\{OPT_C,OPT_E\}$.
\begin{description}
 \item[$\to~\scsi$:] $OPT_C$ is  $\min\{\mathbf{span}(p_i),\mathbf{span}(p_{i+1})\}$. 
 \item[$\tyt~\scsi$:] $OPT_E$ is $\min\{\mathbf{end}(p_i)+dist(p_i,q), \mathbf{start}(p_{i+1})+dist(p_{i+1},q)\}$. 
\end{description}

\begin{theorem}
 Given a colored  point set $P$ on a real line $L$, the smallest color spanning interval containing a query point
 $q$ can be computed in $O(\log n)$ time, using a linear size data structure computed in $O(n\log n)$ time.
\end{theorem}

\remove{
\begin{corollary}
 For any query point $q$, the color spanning vertical corridor (\em corridor of any fixed orientation) containing
 $q$ can be identified in $O(\log n)$ time using a data structure of $O(n)$ time and space complexity after taking
 the projection of each point on $x$-axis in $O(n)$ time.
\end{corollary}}

\section{Smallest color spanning axis-parallel square $(\scss)$} \label{square}\vspace{-0.1in}
In this section, we consider the problem of finding the {\em minimum width color spanning axis parallel square}
($\scss$) among a set $P$ of $n$ points  colored by one of the $k$ possible colors. 
 We first compute the set $\cal S$ of all possible minimal color spanning squares for the given point set $P$ in
 $O(n\log ^2n)$ time, and store them in 
$\Theta(nk)$ space \cite{khanteimouri2013computing}. For a square $s_i \in {\cal S}$, we maintain a five tuple
$(s^\ell_i,s^r_i,s^t_i,s^b_i,\lambda_i)$, where $s^\ell_i$ 
and $s^r_i$ are the $x$-coordinates of the {\em left} and {\em right} boundaries of $s_i$, $s^t_i$ and $s^b_i$
are the $y$-coordinates of its {\em top} and {\em bottom} boundaries 
respectively, and $\lambda_i$ is its side length. 
\subsection{Computation of $\to$ $\scss$} \vspace{-0.1in}
This needs a data structure ${\cal T}_{tb{\ell}r}$ (see Lemma \ref{jaja}) with the point set
$\{(s_i^t,s_i^b,s_i^\ell,s_i^r)$, $i = 1,2,\ldots, n\}$ in $I\!\!R^4$
as the preprocessing. 
Each internal node of the last level of ${\cal T}_{tb{\ell}r}$ contains the minimum length ($\lambda$-value)
of the squares stored in the sub-tree rooted at that node. Given the query point $q$, we use this data structure
to identify the square of minimum side-length in the set ${\cal S}_{contained}$ in
$O\Big((\frac{\log n}{\log \log n})^3\Big)$ time (see Lemma \ref{jaja}), where ${\cal S}_{contained}$ is the subset of 
$\cal S$ that contains $q$.  
\subsection{Computation of $\tyt$ $\scss$}\label{nstab} \vspace{-0.1in}
For the simplicity in the analysis, we further split the squares in ${\cal S}\setminus {\cal S}_{contained}$ 
into two different subcases, namely ${\cal S}_{stabbed}$ and ${\cal S}_{not\_stabbed}$, where ${\cal S}_{stabbed}$ (resp. ${\cal S}_{not\_stabbed}$) 
denote the subset of $\cal S$ that is stabbed (resp. not stabbed) by the vertical line $v(q)$ or the horizontal line $h(q)$ (see Figure \ref{figsquare}).
\begin{figure}
\centering     
\subfigure[{\it Stabbed squares} and {\it Containing squares} with respect to $q$]{\includegraphics[width=45mm]{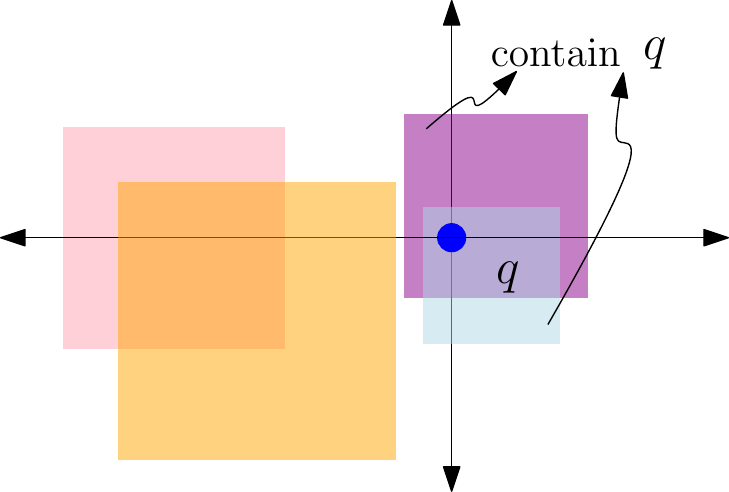}} \hspace{0.2in}
\subfigure[{\it Not Stabbed squares} with respect to $q$]{\includegraphics[width=43mm]{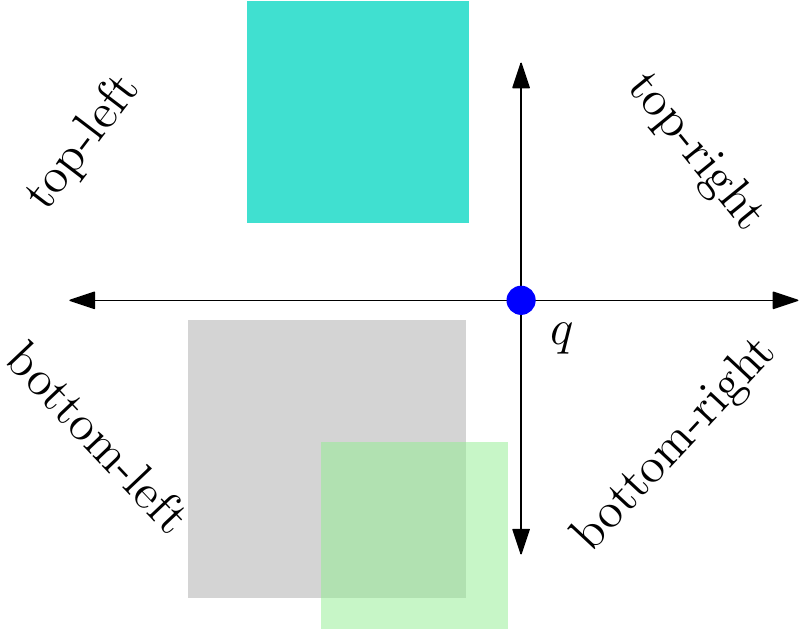}}
\caption{Querying for square}
\label{figsquare}
\end{figure}

We first consider a subset ${\cal S}_{stabbed}^{tbr}$ of the squares in ${\cal S}_{stabbed}$,
whose members are intersected by $h(q)$, and to the left of the point $q$. During the preprocessing,
we have created a data structure ${\cal T}_{tbr}$ with the points $\{(s_i^t,s_i^b,s_i^r), i=1,2,
\ldots, n\}$ in $I\!\!R^3$ (see Lemma \ref{jaja}). Every internal node of the last level of ${\cal T}_{tbr}$ 
contains the maximum value of $s_i^\ell$ coordinate among the points in the sub-tree rooted at that node.
During the query with the point $q$, the search is executed in ${\cal T}_{tbr}$ to find the disjoint
subsets of ${\cal S}_{stabbed}^{tbr}$ satisfying the above query rooted at the third level of the tree.
For each of these subsets, we consider the maximum $s^\ell_i$ value ($x$-coordinate of the left boundary), 
to find a member $s \in {\cal S}_{stabbed}^{tbr}$ whose left boundary is closest to $q$. The square $s$
produces the smallest one among all the squares obtained by extending the members of ${\cal S}_{stabbed}^{tbr}$
so that their right boundary passes through $q$. A similar procedure is followed 
to get the smallest square among all the squares obtained by extending the members of
${\cal S}_{stabbed}^{tb\ell}$ (resp. ${\cal S}_{stabbed}^{{\ell}rb}$, ${\cal S}_{stabbed}^{{\ell}rt}$)
so that their left (resp. bottom, top) boundary passes through $q$, where the 
sets ${\cal S}_{stabbed}^{tb\ell}$, ${\cal S}_{stabbed}^{{\ell}rb}$ and ${\cal S}_{stabbed}^{{\ell}rt}$
are defined as  the set ${\cal S}_{stabbed}^{tbr}$.   

Next, we consider the members in ${\cal S}_{not\_stabbed}$ that lie in the bottom-left quadrant with
respect to the query point $q$. These are obtained by searching the data structure ${\cal T}_{tr}$ with
points $\{(s_i^t,s_i^r), i=1,2,\ldots, n\}$ in $I\!\!R^2$, created in the preprocessing phase. Every 
internal node of ${\cal T}_{tr}$ at the second level contains the $L_\infty$ Voronoi diagram ($VD_\infty$)
of the bottom-left corners of all the squares in the sub-tree rooted at that node.
During the query, it finds the disjoint subsets of ${\cal S}_{not\_stabbed}^{tr}$ satisfying the above query
rooted at the second level of the tree. Now, in each subset (corresponding to an internal node in the second
level of the tree), we locate the closest point of $q$ in $VD_\infty$. The corresponding square produces
the smallest one among all the squares in these sets that are  obtained by extending the members of
${\cal S}_{not\_stabbed}^{tr}$ so that their top/right boundary passes through $q$. The similar procedure
is followed to search for an appropriate element for each of the subsets  ${\cal S}_{not\_stabbed}^{t\ell}$,
${\cal S}_{not\_stabbed}^{br}$ and ${\cal S}_{not\_stabbed}^{b\ell}$, defined as in ${\cal S}_{not\_stabbed}^{tr}$.   

Finally, the smallest one obtained by processing the aforesaid nine data structures is reported. Using
Lemma \ref{jaja}, we have the following theorem:

\remove{
\begin{figure}[t]
\centering
 \includegraphics[scale=0.6]{linftyvd.pdf}\vspace{-0.1in}
 \caption{Different types of bisectors for two points $p,q$ with respect to $L_{\infty}$ distance metric.}\vspace{-0.2in}
 \label{bisectorL}
\end{figure}
}
\begin{theorem}\label{thx}
 Given a colored point set $P$, where each point is colored with one of the $k$ possible colors, for the query point $q$, the smallest color spanning square 
 can be found in $O\Big((\frac{\log N}{\log \log N})^3\Big)$ time, using a data structure built in $O\Big(N(\frac{\log N}{\log \log N})^2\Big)$ time and 
 $O(N\log^2 N)$ space, where $N=\Theta(nk)$.
\end{theorem}
\begin{proof}
 The preprocessing time and space complexity results are dominated by the problem of  finding $\to$ $\scss$ (the smallest square containing $q$). 
The query time for finding (i) $\to$ $\scss$ is  
  $O\Big((\frac{\log N}{\log \log N})^3\Big)$. (ii) $\tyt$ $\scss$ in the set ${\cal S}_{stabbed}$ is $O\Big((\frac{\log N}{\log \log N})^2\Big)$, and 
  (iii) $\tyt$ $\scss$ in the set ${\cal S}_{not\_stabbed}$ is $O(\log^2 N)$ since the search in the 2D range tree\footnote{here the orthogonal range searching result due to \cite{jaja2004space} 
  is not applicable since it works for point set in $I\!\!R^d$, where $d \geq 3$.} is $O(\log N)$, and the search in the $L_\infty$ Voronoi diagram 
  in each of the $O(\log N)$ internal nodes is $O(\log N)$. 

\end{proof}

\section{Smallest color spanning axis-parallel rectangle $(\scsr)$}\label{rectangle}\vspace{-0.1in}
Let us first mention that, we consider the perimeter of a rectangle as its size. Among a set of $n$
points with $k$ different colors, the size of the set $\cal R$ of all possible minimal color spanning
axis parallel rectangles  is $\Theta((n-k)^2)$, and these can be computed in $O(n(n-k)\log k)$ time
\cite{abellanas2001smallest,das2009smallest}. As in Section \ref{square}, (i) we will use
$(s^\ell_i, s^r_i, s^t_i,s^b_i)$ to denote the coordinate of left, right, top and bottom side of
rectangle $R_i \in \cal R$, (ii) define the set $\to$ and $\tyt$ $\scsr$ (see Figure \ref{rectpos}),
 and (iii) split the set $\tyt$ $\scsr$ in $\cal R$ into two subsets ${\cal R}_{stabbed}$ and
 ${\cal R}_{not\_stabbed}$ respectively. 

 \remove{\begin{wrapfigure}{r}{0.45\textwidth}
\vspace{-0.5in}
  \begin{center}
  \includegraphics[scale=0.45]{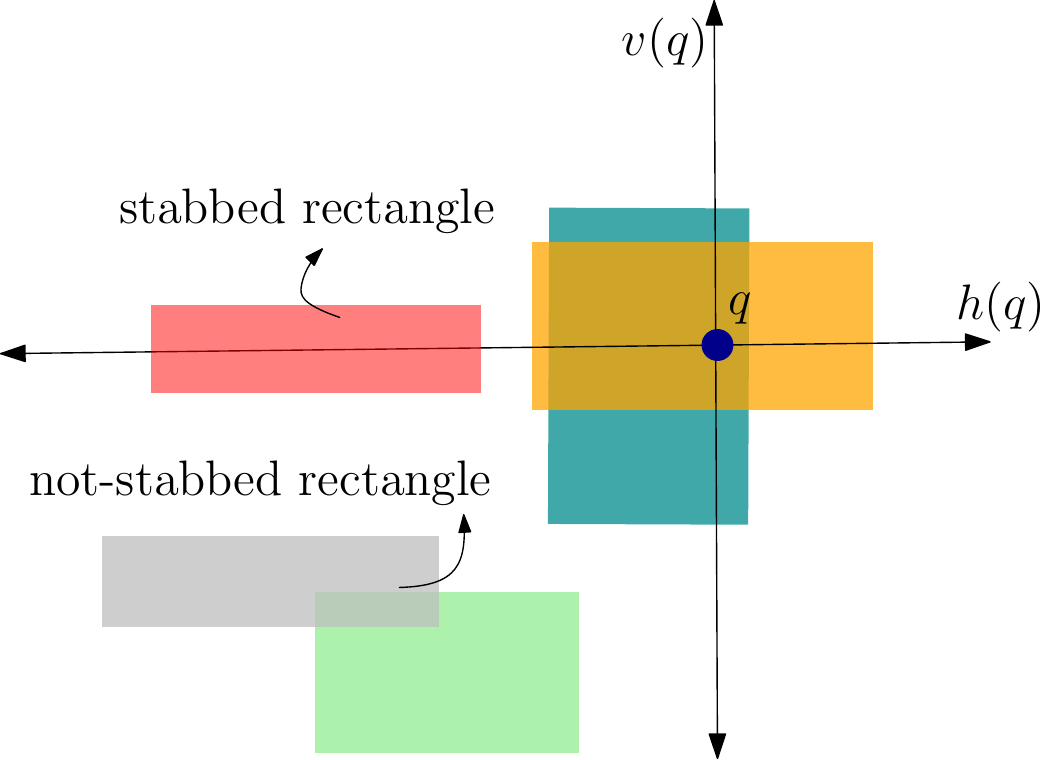}
\end{center}
  \caption{Different possible positions of rectangles with respect to $q$} 
  \vspace{-0.2in}
   \label{rectpos}
   \end{wrapfigure}
}

\begin{figure}[t]
 \centering
 \includegraphics[scale=0.6]{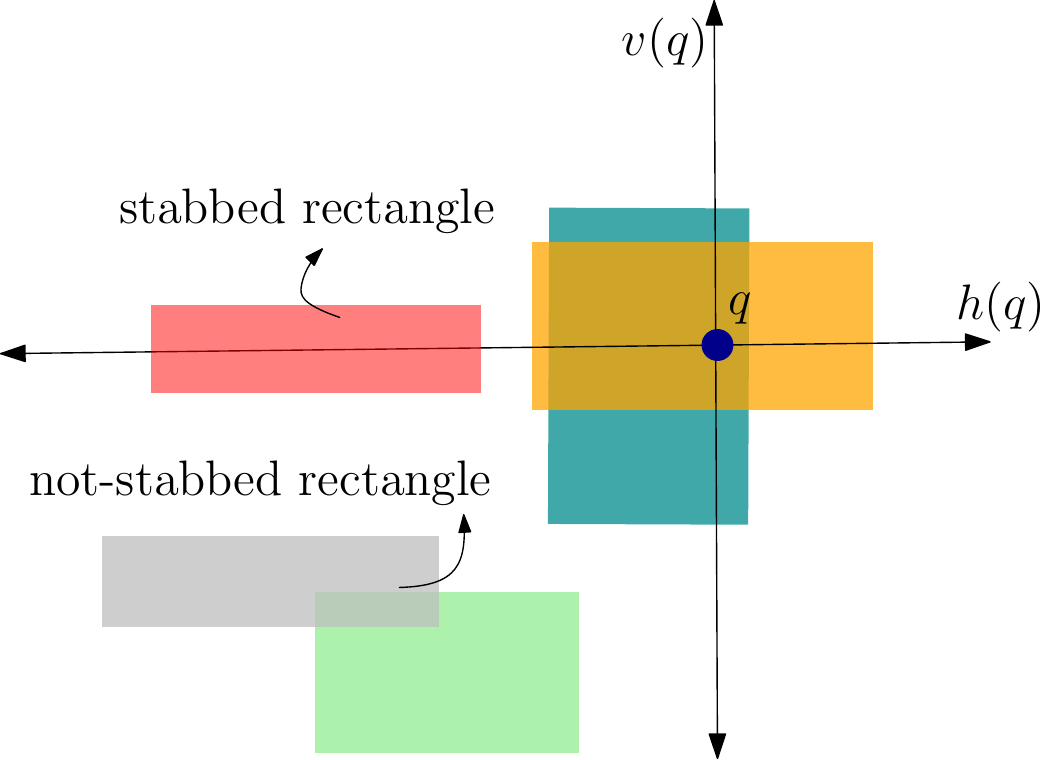}
 \caption{Different possible positions of Rectangles with respect to $q$}
 \label{rectpos}
\end{figure}

The data structures used for finding $\to$ $\scsr$ is exactly the same as that of Section \ref{square};
the complexity results are also the same.  We now discuss the algorithm for computing $\tyt$ $\scsr$.
Let us consider a member $R_i \in {\cal R}_{stabbed}^{tbr}$ ($\subseteq {\cal R}_{stabbed}$), 
where ${\cal R}_{stabbed}^{tbr}$ is the set of rectangles whose top boundaries are above the line
$h(q)$, bottom boundaries are below $h(q)$, and the right boundaries are to the left of $v(q)$. Its
bottom-left coordinate is $(s^\ell_i, s^b_i)$. The size of the rectangle, by extending
$R_i$ so that it contains $q$ on its right boundary, is $(x(q)-s^\ell_i)+h_i = (h_i-s_i^\ell) + x(q)$,
where $h_i=(s^t_i-s^b_i)$ is the height of $R_i$. Thus, to identify the minimum size rectangle by extending 
the members of the set ${\cal R}_{stabbed}^{tbr}$, we need to choose a member of ${\cal R}_{stabbed}$ with the
minimum $(h_i-s_i^\ell)$. Thus, the search structure ${\cal T}_{stabbed}^{tbr}$ for finding the minimum-sized
member in ${\cal R}_{stabbed}^{tbr}$ remains same as that of Section \ref{square}; 
the only difference is that here in each internal node of the third level of ${\cal T}_{stabbed}^{tbr}$,
we need to maintain $\min\{(h_i-s_i^\ell)|R_i ~\text{in that sub-tree}\}$. Similar modifications are done
in the data structure for computing smallest member in each of the sets ${\cal R}_{stabbed}^{tb\ell}$,
${\cal R}_{stabbed}^{{\ell}rt}$ and ${\cal R}_{stabbed}^{{\ell}rb}$.

Next, consider a rectangle $R_i \in {\cal R}_{not\_stabbed}^{tr}$ ($\subseteq {\cal R}_{not\_stabbed}$).
The coordinate of its bottom-left corner is $(s_i^\ell, s_i^b)$. In order to have $q$ on its boundary,
we have to extend the rectangle of $R_i$ to the right up to $x(q)$ and to the top up to 
$y(q)$. Thus, its size becomes $((x(q)-s^\ell_i)+(y(q)-s^b_i)) = ((x(q)+y(q)-(s_i^\ell+s_i^b))$. Hence,
instead of maintaining the $L_\infty$ Voronoi diagram at the internal nodes of the second level of the
data structure ${\cal T}_{not\_stabbed}^{tr}$ (see Section \ref{square}),
we maintain $\max\{(s_i^\ell+s_i^b)|R_i ~\text{in the sub-tree rooted at that node}\}$. This leads
to the following theorem:
\begin{theorem}
Given a set of $n$ points with $k$ colors, the smallest perimeter color spanning rectangle containing
the query point $q$ can be reported in $O\Big((\frac{\log N}{\log \log N})^2\Big)$ time using a data
structure built in $O(N\log N)$  time and $O\Big(N(\frac{\log N}{\log \log N})\Big)$ space,
where $N=\Theta((n-k)^2)$.
 \end{theorem}
\section{Smallest color spanning equilateral triangle $(\scst)$}\vspace{-0.1in}
Now we consider the problem of finding the {\em minimum width color spanning equilateral triangle of fixed orientation} ($\scst$). Without loss 
of generality, we consider that the base of the triangles are parallel to the $x$-axis. For any colored point set $P$ of $n$ points, where each point 
is colored with one of the $k$ possible colors, the size of the set $\triangle$ of all possible minimal color spanning axis parallel triangles is 
$O(n)$ and these triangles can be computed in $O(n \log n)$ time \cite{hasheminejad2015computing}. 

 \remove{\begin{wrapfigure}{r}{0.4\textwidth}
\vspace{-0.1in}
  \begin{center}
  \includegraphics[scale=0.65]{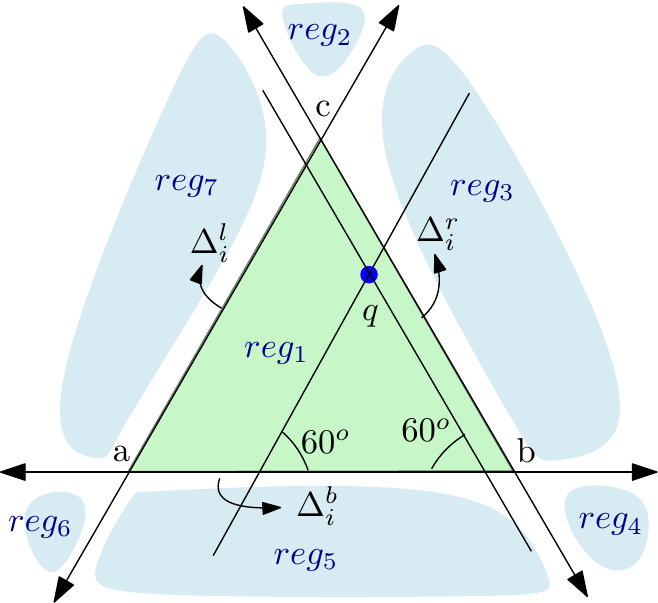}
\end{center}
\vspace{-0.2in}
  \caption{Different regions with respect to a triangle } 
  \vspace{-0.3in}
   \label{triregion}
   \end{wrapfigure}}

   \begin{figure}
     \begin{center}
  \includegraphics[scale=0.8]{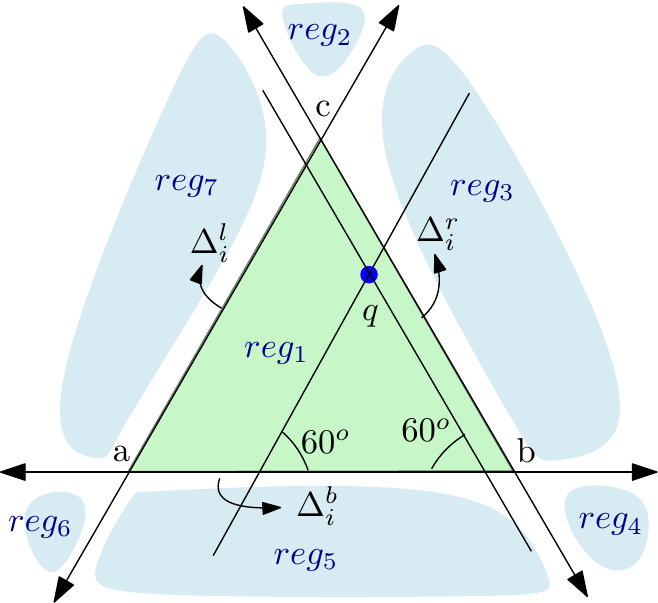}
\end{center}
  \caption{Different regions with respect to a triangle } 
   \label{triregion}
   \end{figure}

As in earlier sections, we will use $s_i^{\ell}$, $s_i^{r}$ and $s_i^{t}$  to denote the {\em bottom-left},
{\em bottom-right}  and {\em top} vertex of a triangle $\triangle_i \in {\cal \triangle}$. We also define 
three lines $\triangle^{l}_i$, $\triangle^{r}_i$ and $\triangle^{b}_i$  as the line containing left, right
and base arms respectively (see Figure 
\ref{triregion}).
 
 \remove{ 
\begin{figure}
 \centering
 \includegraphics[scale=0.8]{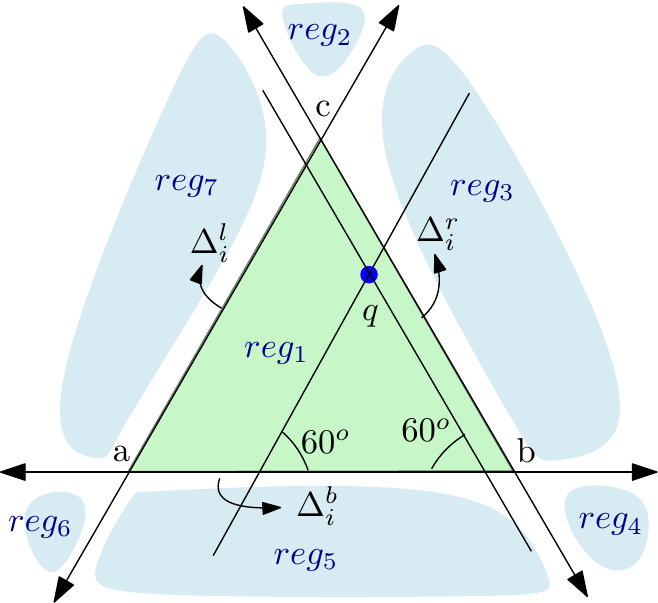}
 \caption{Different regions with respect to a triangle}
 \label{triregion}
\end{figure}
}

\remove{
\begin{observation}\label{rotate}
 Consider a line $\ell$ making an angle $\theta$ with $x$-axis. Also consider a point $p$ lying on or either side of $\ell$. If we rotate the line 
 $\ell$ as well the point $p$ by some angle $\alpha$ with respect to the $x$-axis, then the incidence property remains unchanged, i.e. for the 
 rotated line $\ell'$ the rotated point $p'$ will remain on or on that side of $\ell'$ as it was before rotation (see Figure \ref{triregion}).
\end{observation}
}

We use range query to decide whether the query point $q$ lies inside a triangle $\triangle_i \in {\cal \triangle}$.
For a point $q=(x(q),y(q))$, we consider a six-tuple $(x(q)$, $y(q)$, $x_a(q)$, $y_a(q)$, $x_c(q)$, $y_c(q))$,
where $(x(q),y(q))$ is the coordinates of the point $q$ with respect to normal coordinate axes, say $\nu$-axis, $\big(x_a(q)
=x(q)\cos{60^o}-y(q)\sin{60^o}$, $y_a(q)=x(q)\sin{60^o}+y(q)\cos{60^o}\big)$ are the $x$- and $y$-coordinates 
of $q$ with respect to the $\alpha$-axis, which is the rotation of the axes anticlockwise  by an amount of $60^o$ centered at the origin, 
and $\big(x_c(q)=-x(q)\cos{60^o}-y(q)\sin{60^o}$, $y_a(q)=-y(q)\sin{60^o}+x(q)\cos{60^o}\big)$ are the $x$- and $y$-coordinates 
of $q$ with respect to the $\beta$-axis rotated clockwise by an amount of $60^o$ centered at the origin.

Thus, we can test whether $q$ is inside a triangle $\Delta{abc}$ in $O(1)$ time (see Figure \ref{triregion})
by the above/below test of (i) the point $(x(q),y(q))$ with respect 
to a horizontal line $\overline{ab}$  along the $\nu$-axis, (ii) the point $(x_a(q),y_a(q))$ with respect 
to a horizontal line $\overline{ac}$  along the $\alpha$-axis, and (iii) the point $(x_c(q),y_c(q))$ with respect 
to a horizontal line $\overline{bc}$ along the $\beta$-axis.

\subsubsection{Computation for $\to$ $\scst$}
Here we obtain the smallest triangle among the members in $\cal \triangle$ that contains $q$. In other words, we need to 
identify the smallest triangle among the members in $\cal \triangle$ that are in the same side of
$q$ with respect to the lines $\triangle_i^\ell$, $\triangle_i^r$ and $\triangle_i^b$.
Similar to Section \ref{square}, this needs creation of a range searching data structure ${\cal T}_{b{\ell}r}$ with the points 
$\{(y(s_i^{\ell}),y_a(s_i^\ell),y_c(s_i^r))), i=1,2,\ldots, n\}$ in $I\!\!R^3$ as the preprocessing, where $s^\ell_i, s^r_i$ 
are two vertices of the triangle $\triangle_i \in {\cal \triangle}$ as defined above. Each internal node of the third level of this data 
structure contains the minimum side length ($\lambda$-value) of the triangles stored in the sub-tree rooted at that node. Given 
the query point $q$, it identifies a subset ${\cal \triangle}_{contained}$ in the form of sub-trees in ${\cal T}_{b{\ell}r}$, and during 
this process, the smallest triangle is also obtained. The preprocessing time and space  are $O(n\log n)$ and 
$O(n\frac{\log n}{\log\log n})$, and the query time is $O\Big((\frac{\log n}{\log \log n})^2\Big)$.

\subsubsection{Computation for $\tyt$ $\scst$} 
For each triangle $\tau_i \in {\cal \triangle}\setminus {\cal \triangle}_{contained}$, the lines $\triangle_i^\ell$,
$\triangle_i^r$ and $\triangle_i^b$ defines seven regions, namely $\{\reg_j, j=1,2,\ldots, 7\}$ 
as shown in Fig. \ref{triregion}, where $q \in \tau_i(\reg_i)$ implies $q \in \tau_i$. 
Thus, in order to consider $\tyt$ $\scst$, we need to consider those triangles $\tau_i$ such that $\tau_i\setminus \tau_i(\reg_1)$ regions contain $q$. 

Given the point $q$, we explain the handling of all triangles satisfying $q \in \tau_i(\reg_2)$, and $q \in \tau_i(\reg_3)$ separately
(see Figure \ref{figtriangle}). The cases for $\reg_4$ and $\reg_6$ are handled similar to that of $\reg_2$, and the cases for $\reg_5$
and $\reg_7$ are similar to that of $\reg_3$. The complexity results remain same as for $\to$ $\scst$.

\subsubsection{Handling of triangles with $q$ in $\reg_2$}
Given a query point $q$, a triangle $\tau_i$ is said to satisfy $\reg_2$-condition if $q \in \tau_i(\reg_2)$,
or in other words, $y(q) > y(s_i^\ell) ~\&~ y_a(q) < y_a(s_i^\ell) ~\&~ y_c(q) > y_c(s_i^r)$ (see Figure \ref{figtriangle}. (a)). Let ${\cal \triangle}_2$
be the set of triangles satisfying $\reg_2$-condition with respect to the query point $q$.  
Among all equilateral triangles formed by extending a triangle $\tau_i \in {\cal \triangle}_2$ such that it contains $q$,
the one with top vertex at $q$ will have a minimum size. Thus, among all triangles in ${\cal \triangle}_2 \subset {\cal \triangle}$,
we need to identify the one having a maximum $y(s_i^\ell)$ value. If the internal nodes at the third level of
${\cal T}_{b{\ell}r}$ data structure contains a maximum among $y(s_i^\ell)$ values of the triangles rooted at that node, 
then we can identify a triangle in ${\cal \triangle}_2$ whose extension contains $q$ at its top-vertex and is of minimum size
by searching  ${\cal T}_{b{\ell}r}$ in $O\Big((\frac{\log n}{\log \log n})^2\Big)$ time.

\begin{figure}[h]
\centering     
\subfigure[Handling  triangles with $q \in \reg_2$]{\label{facingtri}\includegraphics[width=45mm]{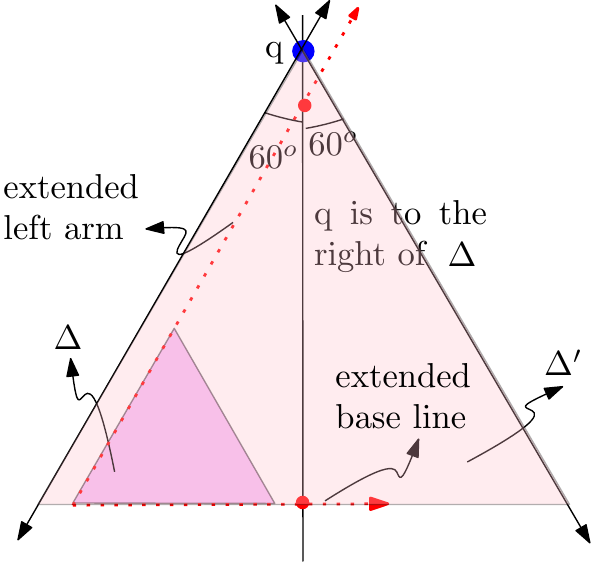}} \hspace{0.5in}
\subfigure[Handling  triangles with $q \in \reg_3$]{\label{facingntri}\includegraphics[width=40mm]{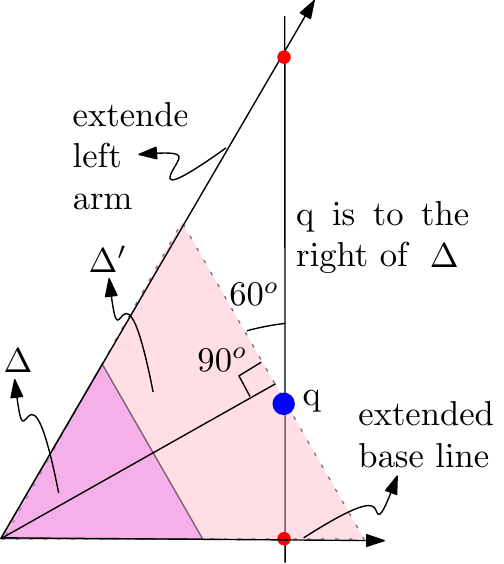}}
\caption{Classification of triangles with respect to $q$}
\label{figtriangle}
\end{figure}

\subsubsection{Handling of triangles with $q$ in $\reg_3$}
Given a query point $q$, a triangle $\tau_i$ is said to satisfy $\reg_3$-condition if $q \in \tau_i(\reg_3)$,
or in other words, $y(q) > y(s_i^\ell) ~\&~ y_a(q) > y_a(s_i^\ell) ~\&~ y_c(q) > y_c(s_i^r)$ (see Figure \ref{figtriangle}. (b)). Let ${\cal \triangle}_3$
be the set of triangles satisfying $\reg_3$-condition with respect to the query point $q$. Among all equilateral triangles
formed by extending a triangle $\tau_i \in {\cal \triangle}_3$ such that it contains $q$, the one whose right-arm passes
through $q$ will have the minimum size. Thus, among all triangles in ${\cal \triangle}_3 \subset {\cal \triangle}$, we need
to identify the one having maximum $y_c(s_i^\ell)$ value. If the internal nodes at the third level of ${\cal T}_{b{\ell}r}$
data structure contains maximum among $y_c(s_i^\ell)$ values of the triangles rooted at that node, then we can identify a triangle
in ${\cal \triangle}_3$ whose extension contains $q$ on its right arm and is of minimum size by searching the ${\cal T}_{b{\ell}r}$
data structure in $O\Big((\frac{\log n}{\log \log n})^2\Big)$ time.

\begin{theorem}
 Given a set of $n$ colored points, the smallest color spanning equilateral triangle of a fixed orientation containing query point $q$ can be reported in  $O\Big((\frac{\log n}{\log \log n})^2\Big)$ time using a data structure of size $O\Big(n(\frac{\log n}{\log \log n})\Big)$, built in 
$O(n\log n)$  time.
 \end{theorem}

\section{Smallest color spanning circle $(\scsc)$} \label{circle}
Here a set $P = P_1 \cup P_2 \cup \ldots \cup P_k$ of points of $k$ different 
colors in $I\!\!R^2$ needs to be preprocessed ($|P_j| \geq 1$), with an aim 
to compute a color-spanning circle ${\cal C}(q)$ of minimum radius containing 
any arbitrary query point $q$. We compute the minimum radius solution among 
two $\scsc$, namely (i) $\to$ $\scsc$ - the minimum radius color-spanning 
circle containing $q$ in its proper interior, and (ii) $\tyt$ $\scsc$ the 
minimum radius color-spanning circle with $q$ on its boundary. To solve 
$\tyt$ $\scsc$, we first consider a constrained version of the problem, where 
a line $\ell$ is given along with the point set $P$, and the query objective 
is to compute the smallest color-spanning circle ${\cal C}_\ell(q)$ of minimum 
radius containing the query point $q$ that is centered on the line $\ell$. We 
solve this constrained problem optimally and use this result to propose an 
$(1+\epsilon)$-factor approximation algorithm for the $\tyt$ $\scsc$ problem 
for a given constant $\epsilon$ satisfying $0 < \epsilon \leq 1$. Thus, our 
proposed algorithm produces an $(1+\epsilon)$-factor approximation solution 
for the $\scsc$ problem.

\subsection{Computation of $\to$ $\scsc$}
In order to compute $\to$ $\scsc$, we create a data structure ${\cal T}_1$. 
We first compute all possible minimal\footnote{A color-spanning circle 
$\chi$ is said to be {\it minimal} if there exists no 
color-spanning circle having a smaller radius that contains the same set of points 
of $k$ colors that are the representative points of $k$ colors in $\chi$.} 
color spanning circles with each point $p_i \in P$, on the boundary. This is 
essentially, the algorithm for computing the smallest color spanning circle 
by Huttenlocher et al. \cite{huttenlocher1993upper}, and is described below 
for the ease of presentation of our proposed method for computing $\to$ 
$\scsc$.

  \begin{figure}
  \centering{\includegraphics[scale=0.9]{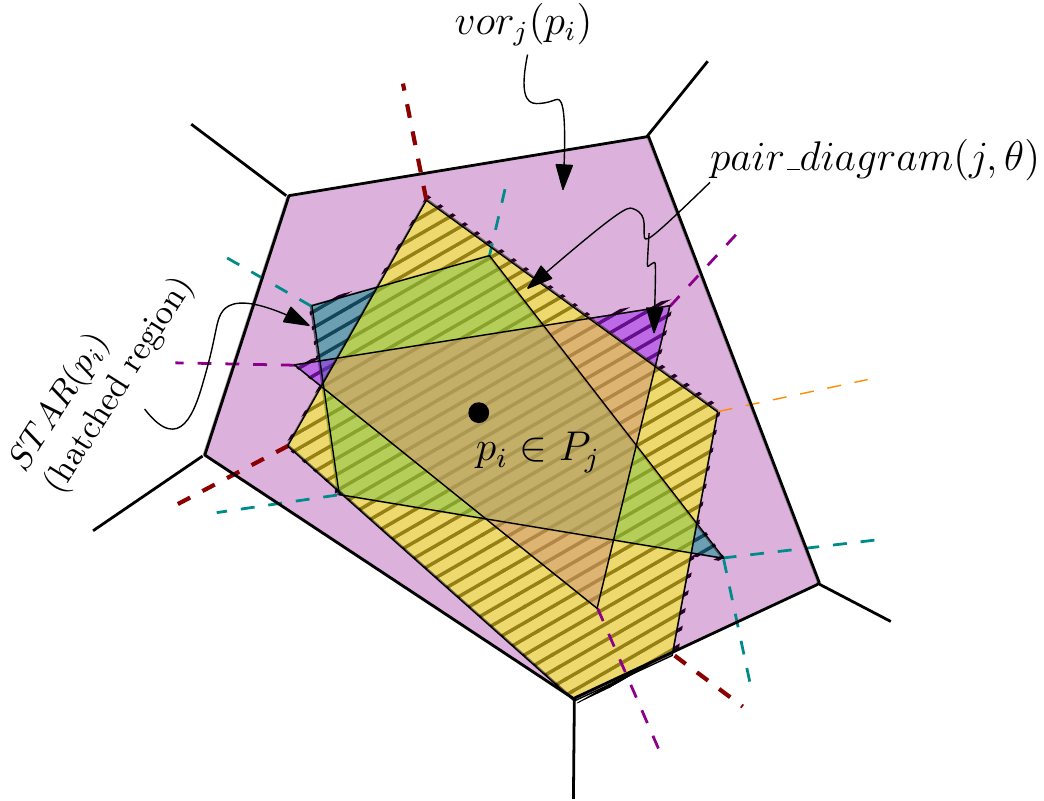}}
    \caption{Star shape polygon (hatched region) for point $p_{i}\in P_j$.}
    \label{fig:vorp}
 \end{figure}

\subsection{Construction of ${\cal T}_1$:}
For each color class $j$, we compute $VOR(P_j)$\footnote{$VOR(Q)$ is the 
Voronoi diagram of a point set $Q$}. We also compute $VOR(P_j\cup P_\theta)$ 
for all $\theta \in \{1,2,\ldots, k\}\setminus\{j\}$, which will be referred 
to as \blue{$\pdia(j,\theta)$}. The Voronoi cell of a point $p_i \in P_j$ in 
$VOR(P_j)$ (resp. $\pdia(j,\theta)$) will be referred to as $vor_j(p_i)$ 
(resp. $vor_{j\theta}(p_i)$). Observe that, for each point $p_i \in P_j$, 
$vor_{j\theta}(p_i)$ will always lie inside $vor_j(p_i)$ for all $\theta 
\in \{1,2,\ldots,k\}\setminus \{j\}$ (see Figure \ref{fig:vorp}). Now, we 
compute the $STAR(p_i) = \bigcup_{\theta \in \{1,2,\ldots,k\}\setminus\{j\}}
vor_{j\theta}(p_i)$, which is a {\it star-shaped polygon} inside $vor_j(p_i)$.
\begin{observation} \label{non-overlap}
\begin{itemize}
\item[(a)] For any color $j \in \{1,2, \ldots, k\}$, \{$STAR(p_i), 
p_i \in P_j$\} are non-overlapping. 
\item[(b)] For a point $p_i \in P_j$, the boundary of $STAR(p_i)$ is the 
locus of the center of the color spanning circles passing through $p_i$ 
(of color $j$) and at least one other point having color $\theta$, where 
$\theta \neq j$.
\end{itemize}
\end{observation}
Thus, we can get the centers $C_i$ of the minimal color-spanning circles 
passing through $p_i$ by drawing the perpendiculars of $p_i$ on the edges 
of $STAR(p_i)$. If the foot of the perpendicular $\pi$ of the point $p_i$ 
on an edge $e$ of  $STAR(p_i)$ lies on $e$ then $\pi$ is included in $C_i$, 
otherwise, the concave vertex $v$ (if any) of $STAR(p_i)$ attached to $e$ 
is included in $C_i$.

Thus, the smallest color-spanning circle in the point set $P$ with 
$p_i \in P$ on its boundary can be obtained by computing all the minimal 
color-spanning circles passing through $p_i$. We compute all possible 
minimal color spanning circles with  each $p_i \in P$, $i= \{1,2,\ldots,n\}$ 
on its boundary. Among a set $P=\{P_1, P_2, \ldots, P_k\}$ of $n$ points in 
$I\!\!R^2$, $P_j$ are the points assigned with color $j$, $|P_j| > 0$ for all 
$j=1,2,\ldots, k$, we use $\cal C$ to denote the set of all possible minimal 
color-spanning circles. In \cite{huttenlocher1993upper}, it is shown that   
${\cal C}=\Theta(nk)$. We now explain the method of preprocessing the 
members in $\cal C$ such that given any query point $q$, the member $\chi \in 
{\cal C}$ of minimum size that contains the point $q$ in its interior can be 
reported efficiently.     

We define a set ${\cal C}'$ of planes in $I\!\!R^3$. Each plane passes through 
the projection of a member of $\cal C$ on the unit paraboloid $\cal U$ centered 
at $(0,0)$. Now, if the projection $q'$ of the query point $q$ on the surface 
of $\cal U$ lies below a plane $\chi' \in {\cal C}'$, then the corresponding 
circle $\chi \in {\cal C}$ contains $q$ in $I\!\!R^2$. Thus, if we assign 
each plane $\chi' \in {\cal C}'$ a weight corresponding to the radius of 
the corresponding disk $\chi \in {\cal C}$, our objective will be to identify the smallest weight plane $\chi'_{min}$ (if any) such that $q'$ lies below 
the plane $\chi'_{min}$.  
 
The dual of the planes in ${\cal C}'$ are points, and the dual of the point 
$q'$ is a plane in $I\!\!R^3$. Thus, in the dual plane, our objective is to 
find the point of minimum weight (if any) that lies in the half-plane 
corresponding to $q'$ containing $(0,0)$. 

We answer this query as follows. We create a height-balanced binary tree 
${\cal T}_1$ with the weights of the dual points (in $I\!\!R^3$) 
corresponding to the planes in ${\cal C}'$. With each internal node 
(including the root) $v$ of ${\cal T}_1$ we attach a data structure 
${\cal T}_1$ ${\cal D}({\cal D}_v)$ with the points ${\cal D}_v$ in the 
subtree of ${\cal T}_1$ rooted at node $v$, that answers the following query:

\begin{itemize}
\item[\ding{237}] Given a set of points $Q$ in $I\!\!R^3$ ($|Q|=m$), ${\cal D}
(Q)$ is a data structure of size $O(m)$ which can be created in time 
$O(m\log m)$, and given any arbitrary half-plane, it can answer whether the 
half-plane is empty or not in $O(\log m)$ time (see Table 3 of 
\cite{agarwal1999geometric}).
\end{itemize}

Given the half-plane corresponding to the query point $q'$ (in $I\!\!R^3$), 
we start searching from the root of ${\cal T}_1$. At a node $v$ in 
${\cal T}_1$, if the search result is ``non-empty'', we proceed to the 
left-subtree of $v$, otherwise we proceed to the right subtree of $v$. 
Finally, we can identify the smallest circle in $\cal C$ containing $q$ 
at a leaf of ${\cal T}_1$. Thus, we have the following result: 

\remove{The query point $q$ is projected on the paraboloid $\cal U$, stays as a $3$d point $q'$ there. In the dual it becomes a half-plane $h$. 

Now we consider the following results for {\it half-plane emptiness query}.

\begin{siderules}
\begin{itemize}
 \item[\ding{237}]  Dobkin et al. \cite{dobkin1990implicitly} showed that, given a convex polyhedra $R$ of $n$ vertices and a half-plane $H$, for any direction $d$, the minimum distance $H$ to be translated in direction $d$ in order to make $R$ and $H$ interior disjoint can be computed in $\log n$ time using a linear size data structure. 
 \item[\ding{237}] Afsani et al. \cite{afshani2009optimal} showed that, given a set $H$ of $n$ planes in $3$d, we can process $H$ into a linear size
 data structure in expected $O(n\log n)$ time such that the lowest plane queries can be answered in $O(\log n)$ time. If the query point $q$ above all
 the planes in $H$, that implies that $q\cap H=\emptyset$; implies $q$ is not inside any circle. On the other hand, if $q$ is below any circle, it 
is below the last circle since for each circle, the portion of the corresponding plane outside the circle is below the unit paraboloid  $\cal U$ 
on which the disks and the point $q$ is projected. Hence the portion of the plane containing the disk is above the paraboloid.
\end{itemize}
\end{siderules}
These lead to the following Lemma;
}

\begin{lemma}\label{empty-query}
 Given the set  $\cal C$ of $O(nk)$ circles, we can preprocess them into a data structure ${\cal T}_1$ in $O(nk\log^2 n)$ time using $O(nk\log n)$ space
 such that for a query point $q$, one can determine the smallest circle in $\cal C$ containing $q$ in  $O(\log^2 n)$ time. 
\end{lemma}

\begin{proof}
As the set of points attached at the nodes in each level of ${\cal T}_1$ are disjoint, the total size of the corresponding $\cal D$ data  structures is $O(nk)$, and the total time 
for their construction is $O(nk\log n)$. Thus, the preprocessing time and space complexity follow. The height of ${\cal T}_1$ is $O(\log n)$, and at each level of ${\cal T}_1$ the aforesaid query needs to be performed at exactly one node. As the half-plane emptiness query complexity is $O(\log n)$, the  complexity result of $\to$ $\scsc$ query follows.    
\end{proof}

\subsection{Constrained version of $\tyt$ $\scsc$ query}\label{cons_cssc}
Before considering $\tyt$ $\scsc$, we consider the following constrained version of 
the $\scsc$ problem:
 \begin{siderules}
  Given a colored point set $P$, a line $\ell$ and a query point $q$, report the smallest color-spanning circle centered
  on the line $\ell$ that contains $q$.
 \end{siderules}

For any point $\rho\in \mathbb{R}^2$, we use $h(\rho)$ to denote the 
horizontal line passing through $\rho$. Without loss of generality, we assume 
that all the points in $P$ are above the line $\ell$. This assumption is 
legitimate since, for an arbitrary point set $P$, we can create another point 
set $P'$ where all the points in $P$ that are above $\ell$ will remain as it 
is, and for each point below $\ell$, we consider its mirror image with respect 
to the line $\ell$. It is easy to observe that, the optimum solution of the 
constrained problem for $P$ and for the transformed point set $P'$ are the same.   

\begin{observation} \label{def1}
A {\em minimal color-spanning circle} among a set of colored points $P$ 
centered on a given line $\ell$ (i) either passes through two points of $P$ 
having different colors, or (ii) passes through a single point in $P$ such 
that the vertical projection of that point on $\ell$ is the radius of the 
said color-spanning circle.   
\end{observation}

Assume the line $\ell$ as the $x$-axis. The squared distance of a point 
$p_i=(\alpha,\beta)$ from the point $(x,0)$ on the $x$-axis is 
$y=d(p_i)=(x-\alpha)^2+\beta^2$ (see Figure \ref{fig:dist_func}), 
which satisfies the equation of a parabola. We call it the {\em distance curve} $f(p_i,\ell)$. 

\begin{figure}[h]
\begin{center}
 \includegraphics[scale=0.8]{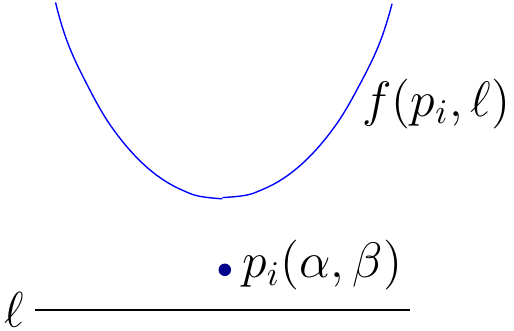}
 \caption{Distance function.}
 \label{fig:dist_func}
\end{center}
\end{figure}

\begin{observation}\label{pseudoline}
For any two distinct points $p_i$ and $p_j$ in $P$, the curves $f(p_i,\ell)$ and $f(p_j,\ell)$ 
intersect exactly at one point.
\end{observation}
\begin{proof}
Follows from the fact that there is exactly one point on the line $\ell$ 
which is equidistant from two points $p_i$ and $p_j$, and it is the point 
of intersection of the perpendicular bisector of $[p_i,p_j]$ and the line 
$\ell$.  
\end{proof}

From now onwards, we will use $P$ to denote the transformed point set and use $P_c \subseteq P$
to denote the set of points of color $c$.  Consider the distance functions $f^c_i$ for all the
points $p_i\in P_c$, and compute their lower envelope ${\cal F}(P_c)$. Thus, for a point 
$\xi$ on the $x$-axis, its nearest point $p_i \in P_c$ corresponds to the function 
$f(p_i,\ell)$ in ${\cal F}(P_c)$ that is hit by the vertical ray drawn at the point $\xi$. The radius of 
the smallest circle centered at $\xi$ that contains at least one point of each 
color is the smallest vertical line segment originating at $\xi$ and hits all the functions 
$\{{\cal F}(P_c), c=1,2,\ldots, k\}$. Thus, the radius of the smallest 
color-spanning circle ${\cal C}(\ell)$ centered on $\ell$ (assumed to be the $x$-axis), 
is the point with minimum $y$-coordinate in the upper envelope ${\cal F}(P)$ of 
all the functions ${\cal F}(P_c), c=1,2,\ldots, k$ (see the dotted curve in Figure 
\ref{fig:upper}).  

\begin{figure}[h]
\vspace{-0.2in}
\begin{center}
 \includegraphics[scale=0.65]{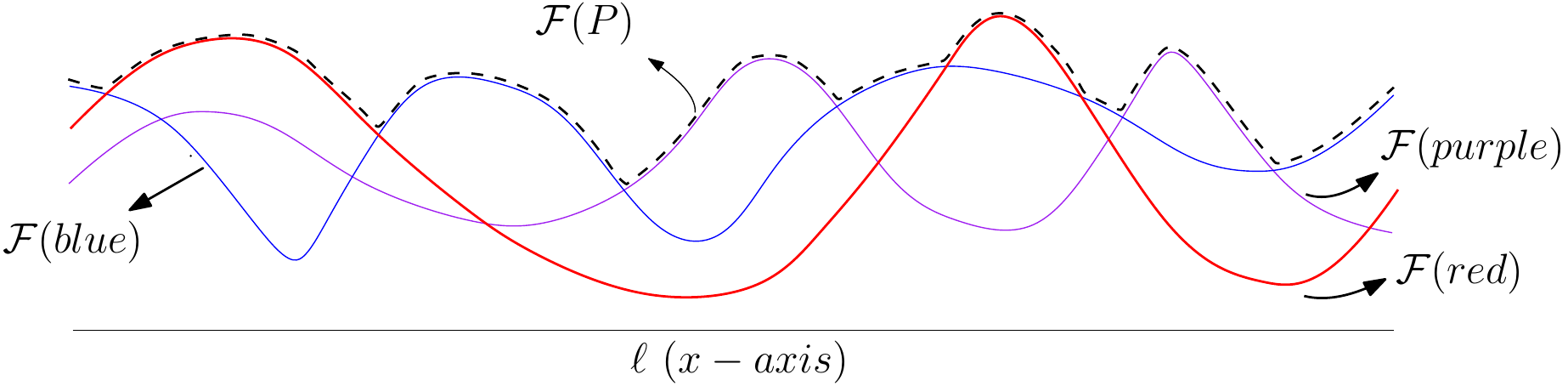}
 \caption{Upper envelope of the lower envelope of all the color classes}
 \label{fig:upper}
\end{center}
\vspace{-0.2in}
\end{figure}

\begin{lemma}
The combinatorial complexity of the function ${\cal F}(P)$ is $O(n\alpha(n))$ 
and can be computed in $O(n \log n)$ time.
\end{lemma}
\begin{proof}
The functions $f(p_i,\ell) \in {\cal F}(P_c)$ behave like pseudolines since each pair 
of functions $f(p_a,\ell), f(p_b,\ell) \in {\cal F}(P_c)$ intersect at exactly one point (see 
Observation \ref{pseudoline}). Thus, the size of their lower envelope is 
$O(n_i)$, and can be computed in $O(n_i\log n_i)$ time 
\cite{vsarir1995davenport}, where $n_i = |P_i|$. 

While computing the upper envelope ${\cal F}(P)$ of ${\cal F}(P_c), c=1,2,
\ldots, k$, consider the parabolic arc segments that are present in all the 
curves $\{{\cal F}(P_c), c=1,2,\ldots, k\}$. These can be treated as 
pseudo-line segments, and as mentioned above, they are $O(n)$ in number. 
Their upper envelope is of size $O(n\alpha(n))$, and can be computed in 
$O(n\log n)$ time \cite{vsarir1995davenport}, where $\alpha(n)$ is the 
inverse of the Ackermann function, which is a very slowly growing function 
in $n$.  
\end{proof}

{\bf Query answering:} During the query, a new point $q$ ($\not\in P$) is 
given, and the objectives are the following:
\begin{siderules}

\begin{itemize}
 \item[1.] compute the smallest color-spanning circle ${\cal C}_1(q,\ell)$ centered on the
 line $\ell$ and containing the point $q$ in its proper interior ({\em constrained $\to$ $\scsc$}), and 
\item[2.] compute the smallest color-spanning circle ${\cal C}_2(q,\ell)$ centered on
the line $\ell$ and the point $q$ on its boundary ({\em constrained $\tyt$ $\scsc$}). 
\end{itemize} 
\end{siderules}

\begin{lemma} \label{two_intersection}
$f(q,\ell)$ intersects ${\cal F}(P)$ in exactly two points, to the left and right 
of the vertical line through $q$ respectively. 
\end{lemma}

\begin{proof}
For a contradiction, let us assume that $f(q,\ell)$ intersects ${\cal F}(P)$ 
more than once to the right of $q$. By Observation \ref{pseudoline}, 
$f(q,\ell)$ can't intersect the same arc-segment of ${\cal F}(P)$ more 
than once. Thus, we consider the case that $f(q,\ell)$ intersects more than 
one arc-segment of ${\cal F}(P)$ to the right of $q$. Let the two consecutive 
points of intersection of $f(q,\ell)$ and ${\cal F}(P)$ to the right side of 
$q$ be with the arc-segments $f(p_i,\ell)$ and $f(p_j,\ell)$ in ${\cal F}(P)$ 
respectively, $i\neq j$. But, observe that, here $f_q$ passes over the vertex 
$v$ on $f(p_i,\ell)$, which is the point of intersection of $f(p_i,\ell)$ and 
the next arc-segment $f(p_{i'},\ell)$, in ${\cal F}(P)$. As both $f(q,\ell)$ 
and $f(p_i,\ell)$ are continuous curves, in order to intersect $f(p_j,\ell)$ 
to the right of the vertex $v$, $f(q,\ell)$ needs to intersect $f(p_i,\ell)$ 
once more. This leads to the contradiction to Observation \ref{pseudoline}.
\end{proof}

As the arc-segments (and hence the vertices) of $\cal F$ are ordered along 
the $x$-axis, we can perform a binary search with $f(q,\ell)$ to find the 
points $a$ and $b$ of its intersections with $\cal F$ to the left and right 
of $q$ respectively. Let $a'$ and $b'$ be the vertical projections of $a$ 
and $b$ respectively on the line $\ell$. We compute the circles passing 
through $q$ with center at $a'$ and $b'$ respectively, and choose the one 
with minimum radius as ${\cal C}_2(q,\ell)$. This needs $O(\log n)$ time in 
the worst case. 

We now explain the computation of ${\cal C}_1(q,\ell)$.  Let $v$ be any point 
on ${\cal F}(P)$ whose vertical projection $v'$ on $\ell$ lies between $a'$ 
and $b'$. Observe that, the line segment $\overline{vv'}$ intersects 
$f(q,\ell)$. Thus, the circle ${\cal C}(v',\ell)$ centered at $v'$ and 
radius $\overline{vv'}$ is color-spanning and contains $q$ in its interior.
Thus the desired ${\cal C}_1(q,\ell)$ (of minimum radius) will be centered 
at some point in the open interval $(a',b')$. Now, we need the following 
definition.

\begin{definition}
Each {\it vertex} of ${\cal F}(P)$ is either (i) the point of intersection 
of consecutive arc-segments in ${\cal F}(P)$, or (ii) the minima of some
arc-segment of ${\cal F}(P)$. 
\end{definition}

\begin{itemize}
\item If there exists a minimal color-spanning circle of type (i) defined by a 
pair of points $p_i,p_j \in P$, then its center is the point of intersection 
of the perpendicular bisector of $[p_i,p_j]$ with the line $\ell$; it can 
also be obtained by the vertical projection of the type (i) vertex, 
generated by the intersection of $f(p_i,\ell)$ and $f(p_j,\ell)$, on the 
line $\ell$ (for a fixed line $=\ell$, both the points are the same). 
\item If there exists a minimal color-spanning circle of type (ii) defined by a 
point $p_i \in P$, then its center is the projection of the
type (ii) vertex of ${\cal F}(P)$ corresponding to $f(p_i,\ell)$.  
\end{itemize}

Thus, in order to get ${\cal C}_1(q,\ell)$, we need to consider all the {\it 
vertices} of ${\cal F}(P)$ whose vertical projections on $\ell$ lie in 
$[a',b']$ and identify the one having a minimum vertical height of ${\cal F}(P)$ 
at that point from $\ell$. In the preprocessing stage, we create a 
height-balanced leaf-search binary tree $\cal T$ with the vertices of 
${\cal F}(P)$ ordered with respect to their $x$-coordinates. Each leaf node 
$v$ is attached with the {\it radius} $|\overline{vv'}|$ of the minimum 
color-spanning circle centered at $v'$. Each interior node of $\cal T$ 
contains the discriminant value of that node (the $x$-coordinate of its 
inorder predecessor) and the {\it minimum} radius stored in the subtree 
rooted at that node. 

At the query time, ${\cal C}_1(q,\ell)$ is obtained by identifying a vertex 
of ${\cal F}(P)$ with minimum attached radius among the vertices whose 
projection lies in $[a',b']$. This can be done in $O(\log n)$ time. Thus, 
we have the following result:
\begin{lemma}\label{res:cons_cssc}
A given set $P$ of colored points and a line $\ell$ can be preprocessed in 
$O(n\log n)$ time and $O(n\alpha(n))$ space such that given any arbitrary 
query point $q$, the minimum radius color-spanning circle centered on $\ell$ 
and containing the point $q$ (on the boundary or in its interior) can be 
computed in $O(\log n)$ time.
\end{lemma} 

\subsection{The general problem}
In this section, we consider a further constrained version of the $\scsc$ 
problem using similar technique as in Section \ref{cons_cssc}, and then use 
it to design the approximation algorithm for the general $\scsc$ problem. 

\subsubsection{Further constrained version}\label{fur_cons_cssc}

\vspace{0.1in}
\begin{siderules}
Given a colored point set $P$, preprocess them such that for a given query point $q$, one can compute the minimum 
radius color-spanning circle ${\cal C}(q,h(q))$ {\em containing $q$ on its
 boundary} which is centered on a horizontal line $h(q)$ passing through $q$. 
\end{siderules}

We propose an optimal algorithm for this problem. For each point $p_i \in P$, 
we construct the data structure ${\cal F}(P,p_i)$, where ${\cal F}(P,p)$ is 
the ${\cal F}(P)$ data structure with the line $\ell=h(p)$. We use ${\cal F}(P,\alpha)$
or ${\cal F}(P,h(\alpha))$ interchangeably to denote the ${\cal F}(P)$ data structure
with respect to line $h(\alpha)$ through a point $\alpha$. Given the query 
point $q$, our objective is to report ${\cal C}(q,h(q))$ by querying in the 
data structure ${\cal F}(P\cup\{q\},q)$. In other words, we need to inspect 
the two points of intersection of the distance curve $f(q,h(q))$ with the 
curve ${\cal F}(P,q)$. Note that, {\it the distance-curve $f(q,h(q))$ is 
composed of two half-lines of slopes ``$+1$'' and ``$-1$'' above the 
horizontal line $h(q)$, originating from $q$.} Our objective is to achieve 
the poly-logarithmic query time. Thus, it is not permissible to construct 
${\cal F}(P\cup\{q\},q)$ during the query. We now describe the method of 
handling this situation.

We identify two consecutive points $p_i, p_{i+1} \in P$ in the sorted order 
of the points in $P$ with respect to their $y$-coordinates satisfying 
$y(p_i) \leq y(q) \leq y(p_{i+1})$, and then use the data structure 
${\cal F}(P,p_i)$ to compute ${\cal C}(q,h(q))$ using the following result.

\begin{lemma}\label{order_same}
For a horizontal line $\ell$ between $h(p_i)$ and $h(p_{i+1})$, the order of 
the points in $P$ with respect to their distances from any point $z \in \ell$ 
remains the same as the order of the members in $P$ with respect to their 
distances from the point $z' \in h(p_i)$, where $z'$ is the vertical 
projection of $z$ on $h(p_i)$.
\end{lemma}

\begin{figure}
 \centerline{\includegraphics[scale=0.5]{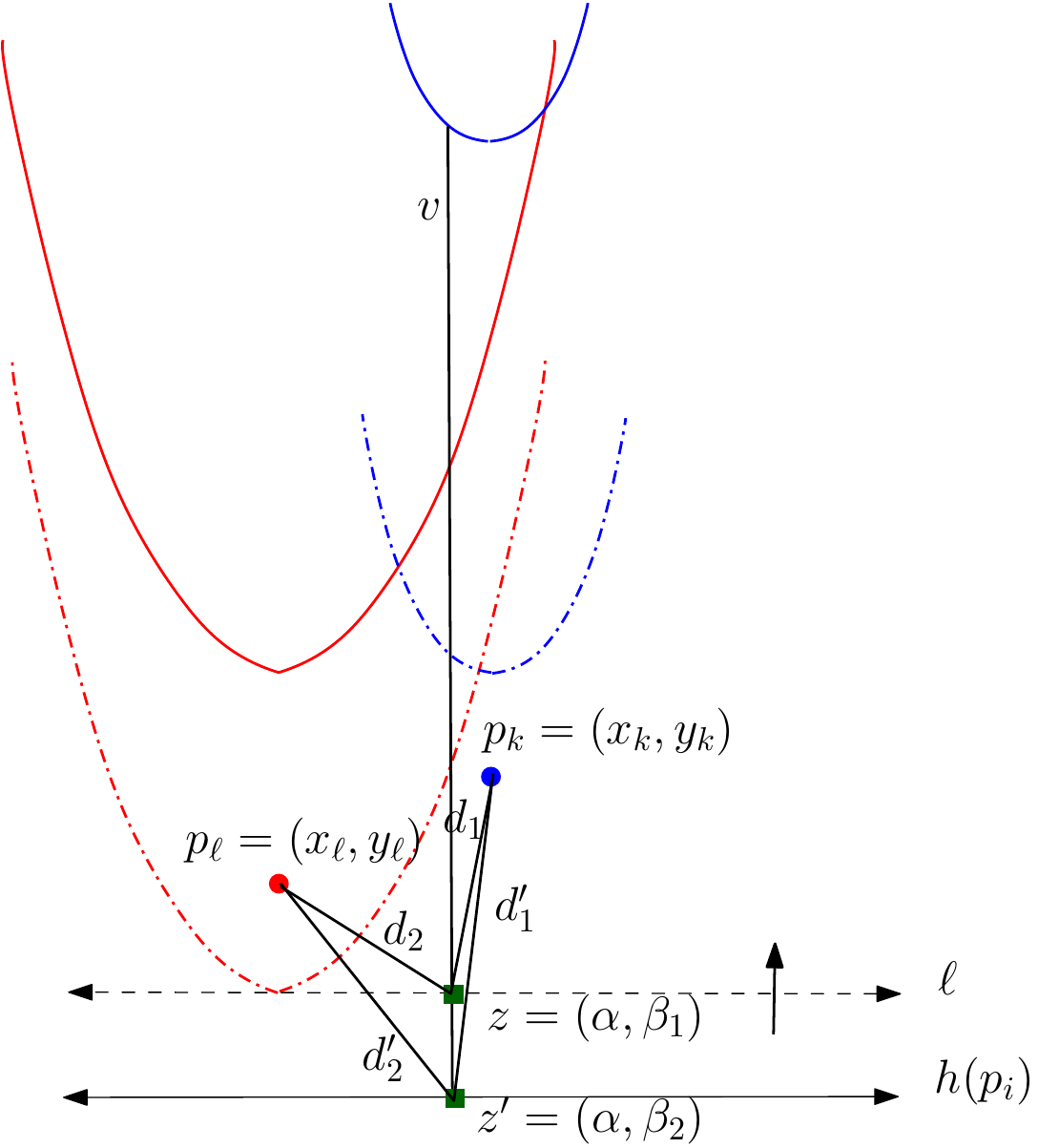}}
 \caption{Illustration of Lemma \ref{order_same}}
 \label{fig:order_same}
\end{figure}

\begin{proof}
We prove this result by contradiction. As the horizontal line $\ell$ lies 
in the horizontal slab bounded by  $h(p_i)$ and $h(p_{i+1})$, the 
coordinates of $z=(\alpha, \beta_1)$ and $z'=(\alpha,\beta_2)$ satisfy 
$\beta_1 > \beta_2$. Let $p_k=(x_k,y_k),p_\ell=(x_\ell,y_\ell) \in P$ be 
two points with $y_k \neq y_\ell$ (see Figure \ref{fig:order_same}). Let $d_1=\delta(p_k,z)$, $d_2=
\delta(p_\ell,z)$, $d_1'=\delta(p_k,z')$ and $d_2'=\delta(p_\ell,z')$, 
where $\delta(.,.)$ is the Euclidean distance between two points.  For 
the contradiction, let $d_1 > d_2$ and $d_1' < d_2'$. We also assume that  
$y_k, y_\ell,\beta_1,\beta_2 > 0 \text{ and } y_k=\min(y_k,y_\ell) > 
\beta_1>\beta_2$. 
\begin{itemize}
\item $d_1 > d_2$ implies that 
$(x_k-\alpha)^2 + (y_k -\beta_1)^2 >(x_\ell -\alpha)^2+(y_\ell -\beta_1)^2$; \\
implying $(y_k -\beta_1)^2 -(y_\ell -\beta_1)^2>(x_\ell -\alpha)^2-
(x_k -\alpha)^2=p$ (say), 
\item whereas $d_1' < d_2'$ implies that, 
$(x_k-\alpha)^2 + (y_k -\beta_2)^2 <(x_\ell-\alpha)^2 + (y_\ell -\beta_2)^2$;\\
implying $(y_k -\beta_2)^2 -(y_\ell -\beta_2)^2<p$.
\end{itemize}

Thus,   $(y_k -\beta_1)^2 -(y_\ell -\beta_1)^2>(y_k -\beta_2)^2 -(y_\ell -\beta_2)^2$,\\ implying 
  $-2\beta_1(y_k-y_\ell)>-2\beta_2(y_k-y_l)$ 
  $\Rightarrow$ $\beta_1<\beta_2$. \\ 
But this contradicts our initial assumption; $\beta_1>\beta_2$.
\end{proof}

Lemma \ref{order_same} says that the distance-curve segments in 
${\cal F}(P,\ell)$  change in a self-parallel manner as $\ell$ moves 
from $p_i$ to $p_{i+1}$. In other words, (i) the representative point 
of each color ($P_i$) lying in the smallest color-spanning circle 
remains same if its center $z$ is any point on the vertical line 
$\overline{z'z''}$, $z' \in h(p_i)$, $z'' \in h(p_{i+1})$, and 
$\overline{z'z''} \perp h(p_i)$, and (ii) the order of the points inside 
the smallest color-spanning circle with respect to their distances from 
the center $z$ remain same for all the points $z \in \overline{z'z''}$.

\begin{figure}[h]
\begin{center}
 \includegraphics[scale=0.5]{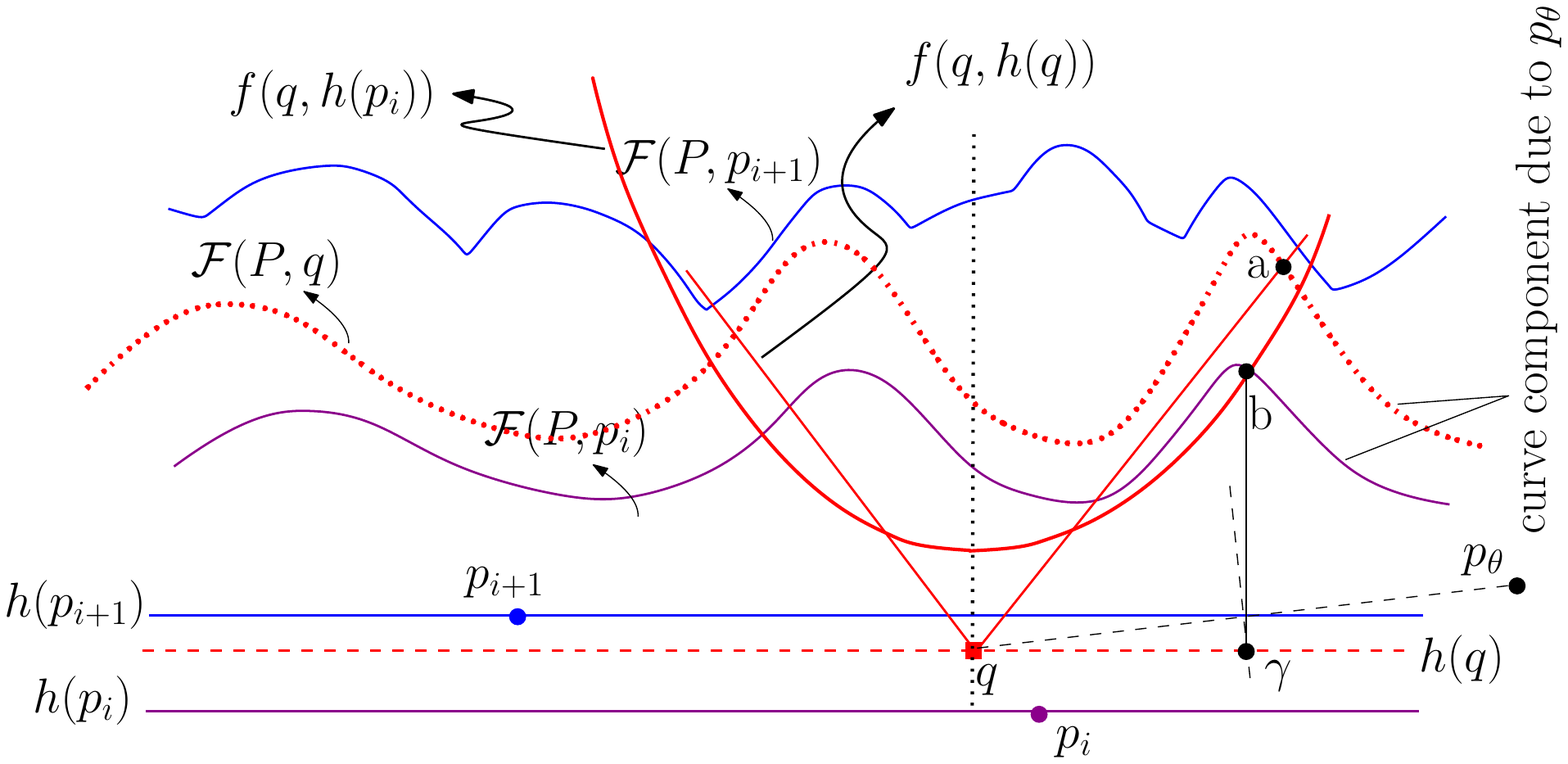}
 \caption{Schematic description of the generation of $\gamma$, the center of $C_1$: here $q$ satisfies $y(p_i)\leq y(q)\leq y(p_{i+1})$.}
 \label{fig:h(q)}
\end{center}
\end{figure}
 
Moreover, if the distance-curve $f(q,h(q))$ (resp. $f(q,h(p_i))$) intersects 
${\cal F}(P,q)$ (resp. ${\cal F}(P,p_i)$) at the point $a$, (resp. $b$) to the 
right side of $q$, then both the points $a$ and $b$ lie on the curve-component 
of the same point (say $p_\theta\in P$) in ${\cal F}(P,q)$ and ${\cal F}(P,p_i)$ 
respectively. Thus, both the smallest color-spanning circles passing through 
$q$ and centered on $h(q)$ and $h(p_i)$ respectively also pass through 
$p_\theta$. Thus, we can identify $f(p_\theta,h(q))\in {\cal F}(P,q)$ by 
performing binary search in ${\cal F}(P,p_i)$ with $f(q,h(p_i))$. The center 
of the smallest color-spanning circle $C_1$ passing through $q$ and centered 
on the line $h(q)$, is the point of intersection $\gamma$ of the perpendicular 
bisector of $p_\theta$ and $q$ and the line $h(q)$; its radius is  
$|\overline{q\gamma}|$ (see Figure \ref{fig:h(q)}). Similarly, we can compute another point $\gamma'$ 
on $h(q)$ to the left of $q$ such that the circle $C_2$ centered at $\gamma'$ 
and radius $\overline{\gamma'q}$ is also a minimal color-spanning circle 
passing through $q$. ${\cal C}_2(q,h(q))$ is the smallest one among $C_1$ and 
$C_2$. Thus,

\remove{
\begin{figure}[h]
\begin{center}
 \includegraphics[scale=0.5]{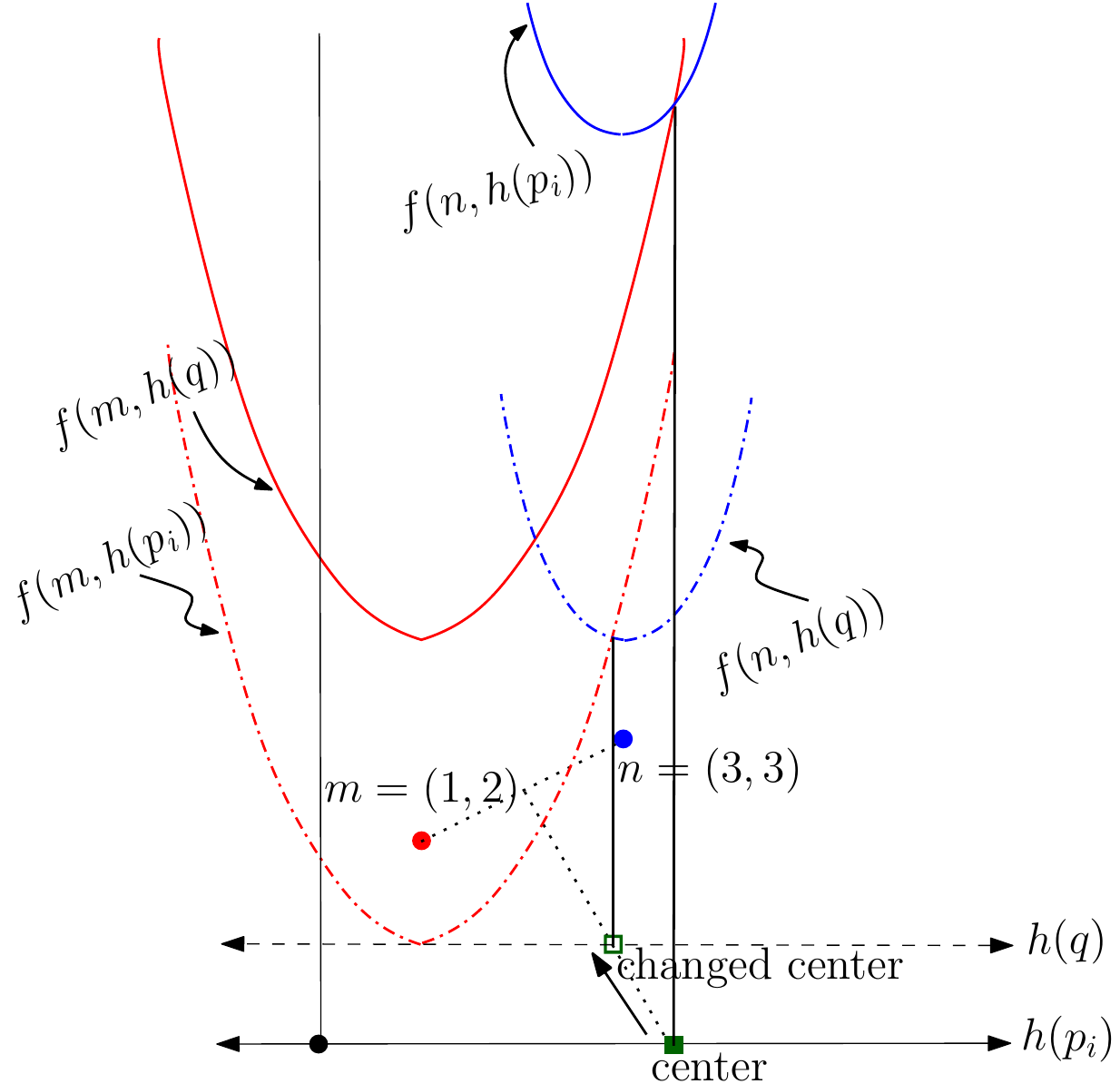}
 \caption{The center moves along the perpendicular bisector of $m$ and $n$.}
 \label{fig:dc_change}
\end{center}
\end{figure}
}

\begin{lemma}
A colored point set $P$ of size $n$, where each point has one of the $k$ given colors, can be preprocessed 
in $O(n^2\log n)$ time and $O(n^2\alpha(n))$ space such that for a query point $q$, the $\tyt$ $\scsc$ with 
center on $h(q)$ can be computed  in $O(\log n)$ time. 
\end{lemma}

\subsection{Unconstrained $\tyt$ $\scsc$}\label{uncons_cssc}
We now present an $(1+\epsilon)$-approximation algorithm for the unconstrained 
version of the $\tyt$ $\scsc$ problem. We draw the lines ${\cal L} = \{\ell_1, 
\ell_2, \ldots, \ell_{\frac{\pi}{\epsilon}}\}$ through the origin $(0,0)$ such 
that two consecutive lines make an angle $\epsilon$ at the origin (see Figure \ref{fig:c_orient}). For each of 
these lines $\ell_i$ as the $x$-axis, we construct the data structure 
${\cal D}\equiv {\cal F}(P)$ with the points in $P$, as described in 
Section \ref{cons_cssc}. Given the query point $q$, we can compute the 
constrained $\tyt$ $\scsc$ with center on the line $\ell_i(q)$ passing through 
$q$ and parallel to the line $\ell_i \in {\cal L}$ for each $i \in \{1,2,
\ldots, \frac{\pi}{\epsilon}\}$. Finally, report the one of minimum size, 
namely ${\cal C}_2(q)$. We now show the following:
\remove{
\begin{figure}[h]
\begin{center}
 \includegraphics[scale=0.6]{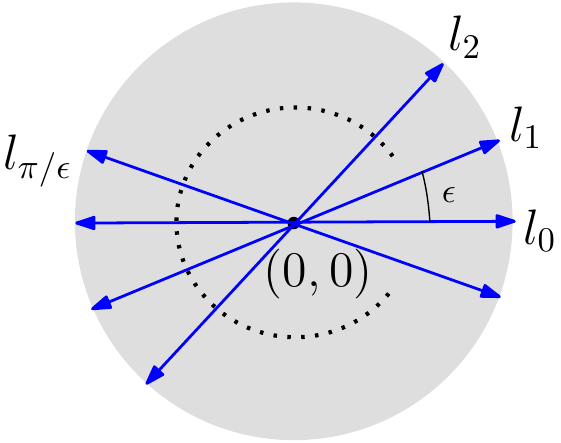}
 \caption{Dividing the plane into several orientations.}
 \label{fig:c_orient}
\end{center}
\end{figure}
}

\begin{figure}[h]
\begin{minipage}[b]{0.45\linewidth}
\centering
\centerline{\includegraphics[scale=0.6]{orientations}}
 \caption{Dividing the plane into several angular orientations.}
 \label{fig:c_orient}
\end{minipage} %
\begin{minipage}[b]{0.55\linewidth}
\centering
\centerline{\includegraphics[scale=0.6]{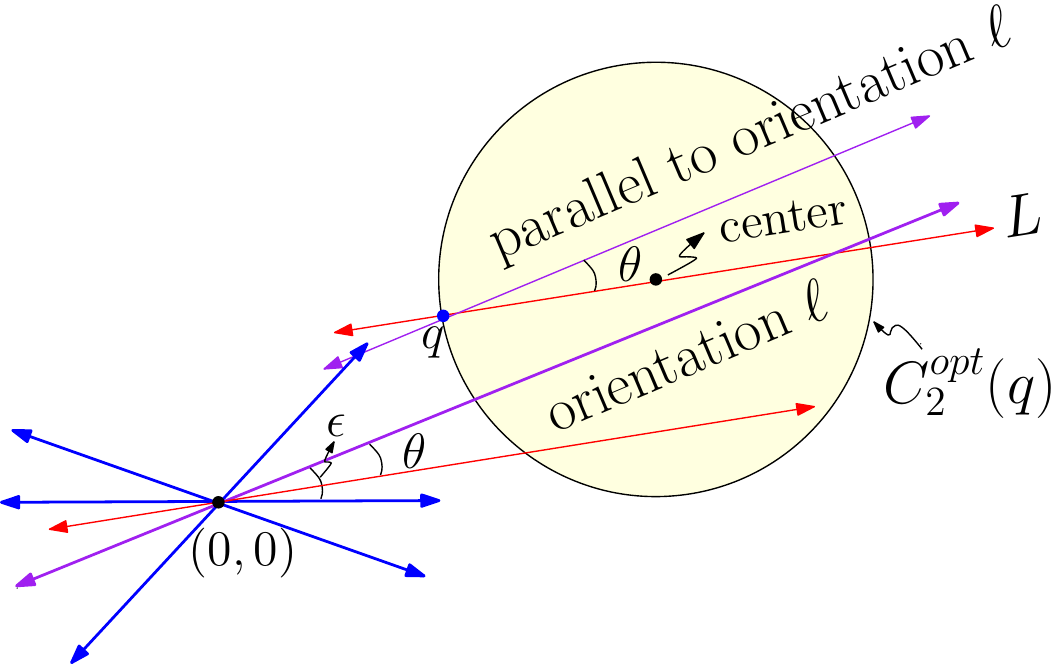}}
\caption{The circle ${\cal C}_2^{opt}(q)$.}
 \label{fig:l_orient}
\end{minipage} %
\end{figure}

\begin{lemma}\label{res:approx-scsc}
$\displaystyle\frac{radius({\cal C}_2(q))}{radius({\cal C}_2^{opt}(q))} \leq 1+\epsilon$, where ${\cal C}_2^{opt}(q)$
is the optimum unconstrained $\tyt$ $\scsc$ passing through $q$. 
\end{lemma}

\begin{proof}
 Without loss of generality assume that $radius({\cal C}_2^{opt}(q))=1$, 
 and its center lies on the line $L$ passing through $q$. 
 Observe that, we have reported ${\cal C}_2(q,\ell_i(q))$ for some line 
 $\ell_i(q)$ which makes an angle $\theta$ with the line $L$ where
 $\theta \leq \epsilon$ (see Figure \ref{fig:l_orient}).
We justify that $radius({\cal C}_2(q,\ell(q)))\leq 1+\epsilon$ by showing that there exists a color spanning circle $C$ with center on $\ell(q)$ and radius 
$\rho \leq 1+ \epsilon$. We construct the circle $C$ as follows (see Figure \ref{fig:construct}):
 
\remove{
 \begin{figure}[h]
\begin{center}
 \includegraphics[scale=0.6]{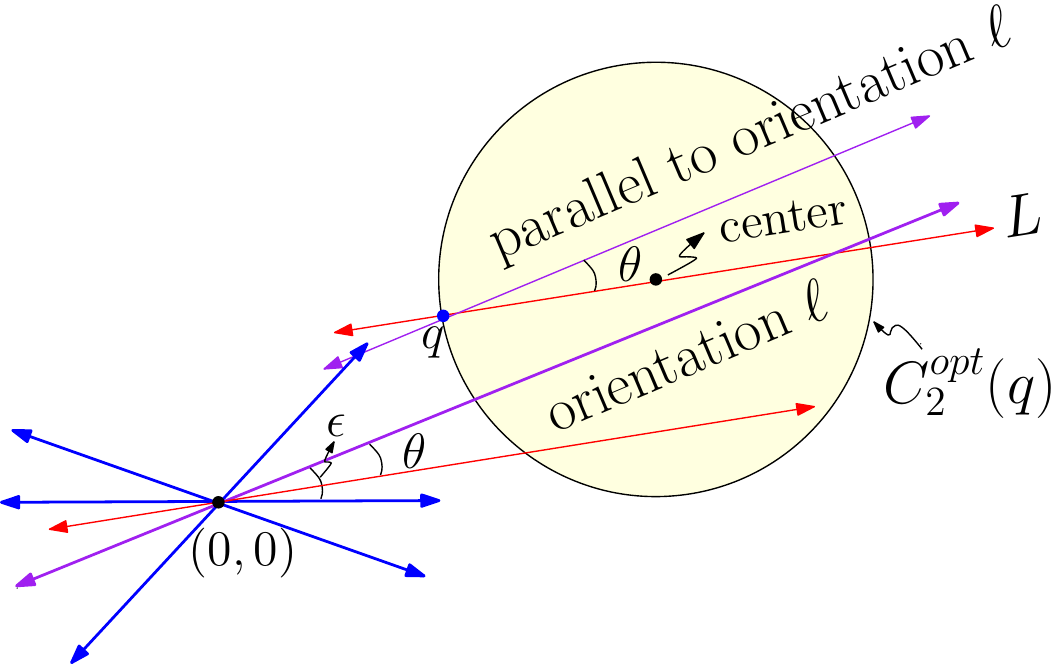}
 \caption{The circle ${\cal C}_2^{opt}(q)$.}
 \label{fig:l_orient}
\end{center}
\end{figure}
}

\begin{figure}[htbp]

\begin{minipage}[b]{0.45\linewidth}
\centering
\centerline{\includegraphics[scale=0.75]{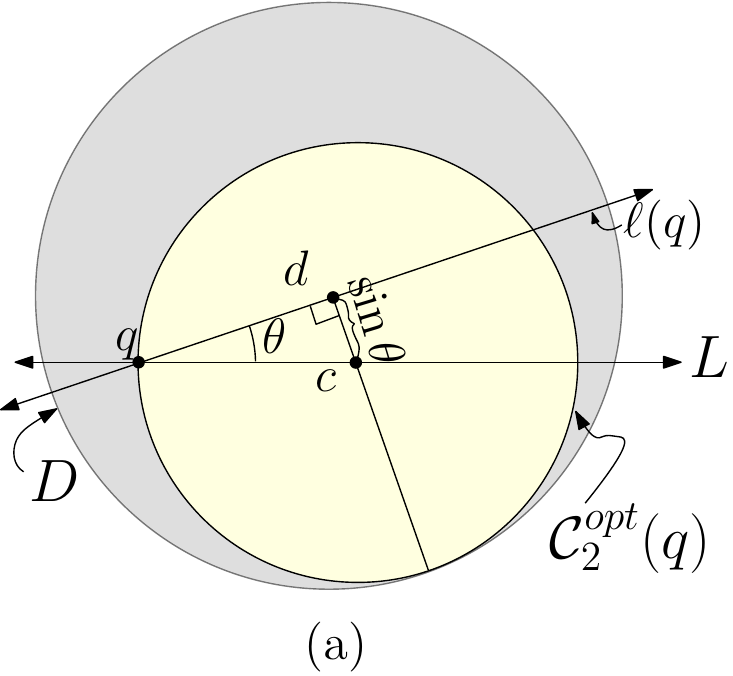}}
\end{minipage} %
\begin{minipage}[b]{0.55\linewidth}
\centering
\centerline{\includegraphics[scale=0.9]{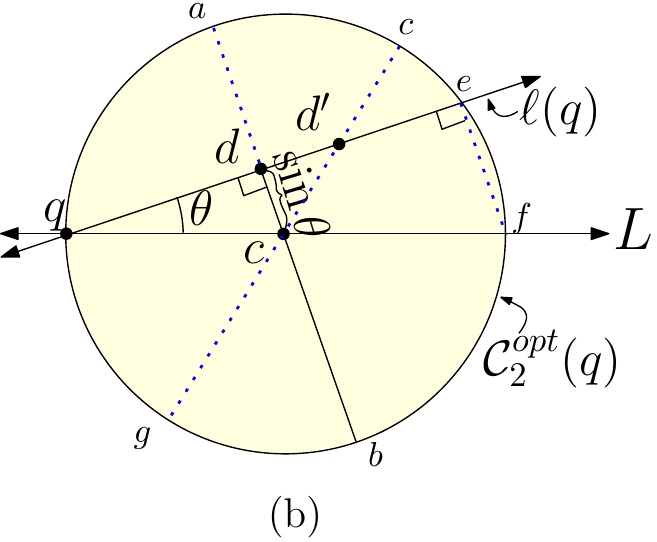}}
\end{minipage} %
\caption{Construction of $D$} 
\label{fig:construct}
\end{figure}

Draw a perpendicular line from the center $c$ of the circle ${\cal C}_2^{opt}(q)$ on the line $\ell(q)$ (parallel to $\ell_i \in {\cal L}$). 
Let it intersect the line $\ell(q)$ at point
$d$ (see Figure \ref{fig:construct} (a)). Draw a circle $D$ centered at $d$ with radius $1+\sin\theta$.
Observe that,
\begin{enumerate}
 \item $\mathbf{D}$ {\bf contains} $\mathbf{{\cal C}_2^{opt}(q)}$: since $radius(D)=1+\sin\theta$, and the construction shows that the 
 radius of $D$ subsumes the radius of ${\cal C}_2^{opt}(q)$. 
 \item $\mathbf{radius(D) = 1+\sin\theta ~\leq~ 1+\theta~ \leq~1+\epsilon}$: since $\epsilon$, and hence $\theta$, is very small (i.e., $\epsilon \longrightarrow 0$). 
 \item $\mathbf{D}$ {\bf is the smallest circle containing} $\mathbf{{\cal C}_2^{opt}(q)}$ {\bf centered on} $\mathbf{\ell(q)}$: This can be proved by contradiction.
 Let us consider that there is another circle $D'$ centered at the point $d'(\neq d)$ on the line $\ell(q)$ that contains
 ${\cal C}_2^{opt}(q)$, and have a smaller radius than $D$ (see  Figure \ref{fig:construct}(b)). We need to show that
 $radius(D') = r''= \overline{d'g}> 1+\sin\theta$.  We have $\overline{d'c}>\overline{dc} = \sin\theta$.
 Thus, $r''=1+\overline{d'c} > 1+\sin\theta = r' = radius(D)$.
\end{enumerate}
Hence, the claim follows.
\end{proof}

From Lemma \ref{empty-query}, \ref{res:cons_cssc} and \ref{res:approx-scsc}, we have the following result: 
\begin{theorem}
Given a set of $n$ colored points, the $(1+\epsilon)$ smallest color spanning circle containing query point $q$ can be reported in $O(\log n)$ time, using a data structure of size $O(n^2\alpha(n))$, built in $O(n^2k \log n)$ time.
\end{theorem}

\section{Conclusion}
In this paper, we studied the color spanning versions of various localized query problems. To the best of our knowledge, this color spanning variation of localized query problem has not been studied yet. This type of problems has a lot of applications in real life, especially in the {\em facility location}. We hope  this will attract a lot of researchers to study further variations of this problem. For the query version of the $\scsc$ problem, obtaining an exact solution in sublinear query time is an open problem.

\subsection*{Acknowledgment:} The authors acknowledge an important suggestion given by  Michiel Smid for solving the $\scsc$ problem. 

\bibliography{mybibfile}

\end{document}